\documentclass{llncs}
\usepackage{makeidx}  

\newcommand{\SKIP}{\mathtt{skip}}
\newcommand{\EMPTY}{\mathtt{empty}}

\newcommand{\PASSIGN}[2]{{#1}\mathrel{:\approx}{#2}}
\newcommand{\DIVERGE}{\mathtt{diverge}}
\newcommand{\ITE}[3]{\mathtt{if} \left( {#1} \right) \left\{ {#2} \right\} \mathtt{else} \left\{ {#3} \right\}}

\newcommand{\COMPOSE}[2]{{#1};\:{#2}}

\newcommand{\WHILE}[1]{\mathtt{while} \left( {#1}\right) \{ }
\newcommand{\WHILEDO}[2]{\mathtt{while} \left( {#1}\right) \left\{ {#2} \right\}}

\newcommand{\IF}[1]{\mathtt{if}\left({#1}\right)\left\{\right.}
\newcommand{\IFBRACK}[1]{\mathtt{if} \left( {#1} \right) \{}

\newcommand{\ELSE}{~\mathtt{else}~}
\newcommand{\REPEAT}{\mathtt{repeat}~\{}
\newcommand{\UNTIL}[1]{\}~\mathtt{until}\left({#1}\right)}
\newcommand{\REPEATUNTIL}[2]{\mathtt{repeat}\left\{{#1}\right\}\mathtt{until}\left( {#2}\right)}

\newcommand{\To}{\rightarrow}

\newcommand{\iverson}[1]{\left[ {#1} \right]}

\newcommand{\wpsymbol}{\mathsf{wp}}
\renewcommand{\wp}[2]{\wpsymbol\left\llbracket{#1}\right\rrbracket\left( {#2} \right)}

\newcommand{\ertsymbol}{\mathsf{ert}}
\newcommand{\boldertsymbol}{\textsf{\textbf{ert}}}
\newcommand{\ert}[2]{\ertsymbol\left\llbracket{#1}\right\rrbracket\left( {#2} \right)}

\newcommand{\subst}[2]{[{#1}/{#2}]}
\newcommand{\lfp}{\textup{\textsf{lfp}}~}

\newcommand{\statesubst}[2]{\left[ {#1} \mapsto {#2}\right]}

\newcommand{\Exp}[3]{\int_{#1}~{#2}~d{#3}}
\newcommand{\semantics}[2]{\left\llbracket {#1} \right\rrbracket_{#2}}

\newcommand{\RelSymbolNewRule}{\mathrel{\not\mathrel{\Cap}}}
\newcommand{\RelNewRule}[2]{{#1} \RelSymbolNewRule {#2}}

\newcommand{\Nats}{\mathbb{N}}
\newcommand{\Reals}{\mathbb{R}}
\newcommand{\Rpos}{\mathbb{R}_{\geq 0}}

\newcommand{\Rposinf}{\mathbb{R}_{\geq 0}^{\infty}}
\newcommand{\E}{\mathbb{E}}
\newcommand{\Vars}{\textsf{Vars}}
\newcommand{\VarsInExp}[1]{\textsf{Vars}\left( #1 \right)}

\newcommand{\VarsAssign}[1]{\textsf{Mod} \left( #1 \right)}
\newcommand{\VarsGuard}[1]{\textsf{Vars}_{\textsf{Guard}} \left( #1 \right) }

\newcommand{\Vals}{\textsf{Vals}}
\newcommand{\Rats}{\mathbb{Q}}
\newcommand{\States}{\Sigma}
\newcommand{\pgcl}{\textnormal{\sfsymbol{pGCL}}\xspace}   
\newcommand{\bnl}{\textnormal{\sfsymbol{BNL}}\xspace}   
\newcommand{\Dists}[1]{\mathcal{D}\left( {#1} \right)}

\newcommand{\BN}{\mathcal{B}}
\newcommand{\EBN}{\mathcal{B}}

\newcommand{\EBNS}{\text{EBN}}
\newcommand{\NODES}{V}
\newcommand{\INPUTS}{I}
\newcommand{\EDGES}{E}
\newcommand{\VALUES}{\Vals}
\newcommand{\DEP}{\mathsf{dep}}
\newcommand{\CPTSYM}{\mathsf{cpt}}
\newcommand{\CPT}[1]{\CPTSYM[#1]}

\newcommand{\PROB}[1]{\mathsf{Pr}\left(#1\right)}

\newcommand{\T}[1]{\mathbf{#1}}
\newcommand{\EMPTYT}{\mathbf{\varepsilon}}
\newcommand{\El}[2]{#1(#2)}

\newcommand{\BVAL}[1]{\underline{#1}}

\newcommand{\BNCOND}{\textit{cond}}

\newcommand{\TGUARD}[3]{\textit{guard}_{#1}(#2,#3)}
\newcommand{\TASSIGN}[3]{\textit{assign}_{#1}(#2,#3)}
\newcommand{\TBLOCK}[2]{\textit{block}_{#1}(#2)}
\newcommand{\TNODE}[3]{\textit{block}_{#1}(#2,#3)}
\newcommand{\TBN}[2]{\textit{BNL}(#1,#2)}
\newcommand{\TB}[1]{\textit{BNL}(#1)}

\newcommand{\TVAR}[1]{x_{#1}}

\newcommand{\ROOTS}[1]{\textit{roots}(#1)}

\newcommand{\true}{\textsf{true}}
\newcommand{\false}{\textsf{false}}
\newcommand{\sfsymbol}[1]{\textsf{\upshape {#1}}}
\newcommand{\mydot}{\text{{\huge\textbf{.}}~}}

\newcommand{\qiff}{\quad\textnormal{iff}\quad}
\newcommand{\qqiff}{\qquad\textnormal{iff}\qquad}

\newcommand{\ppreceq}{~{}\preceq{}~}

\newcommand{\eeq}{~{}={}~}
\newcommand{\nneq}{~{}\neq{}~}
\newcommand{\lleq}{~{}\leq{}~}

\newcommand{\nodeS}{S}
\newcommand{\nodeR}{R}
\newcommand{\nodeG}{G}

\newcommand{\setcomp}[2]{\left\{\, {#1} ~\middle|~ {#2} \,\right\}}

\newcommand{\valueIn}[1]{\langle #1 \rangle}

\usepackage{amssymb, amsmath}
\usepackage{hyperref}
\usepackage{nicefrac}
\usepackage{xfrac}
\usepackage{tensor}
\usepackage{mathtools}
\usepackage{subfigure}
\usepackage[textwidth=34mm]{todonotes}
\usepackage{stmaryrd}
\usepackage{xspace}
\usepackage{adjustbox}

\usepackage{tikz}
\usetikzlibrary{fit,calc,positioning,decorations.pathreplacing,matrix,shapes, arrows}
\usepackage{pgfplots}
\usetikzlibrary{calc}
\usepackage{colortbl}
\colorlet{shadegray}{gray!40}

\allowdisplaybreaks

\begin{document}
\mainmatter              
\title{How long, O Bayesian network,\\ will I sample thee?}
\subtitle{A program analysis perspective on\\ expected sampling times.\\[1ex](arXiv extended version)}                     
\titlerunning{How long, O Bayesian network, will I sample thee?}  
%
\author{Kevin Batz \and Benjamin Lucien Kaminski \and \\Joost-Pieter Katoen \and Christoph Matheja}
\authorrunning{Batz, Kaminski, Katoen, \& Matheja} 
%
\tocauthor{Kevin Batz, Benjamin Lucien Kaminski, Joost-Pieter Katoen, Christoph Matheja}
\institute{RWTH Aachen University}

\maketitle              

\begin{abstract}
Bayesian networks (BNs) are probabilistic graphical models for describing complex joint probability distributions.
The main problem for BNs is inference:
Determine the probability of an event given observed evidence.
Since exact inference is often infeasible for large BNs, popular approximate inference methods rely on sampling.

We study the problem of determining the expected time to obtain a single valid sample from a BN.
To this end, we translate the BN together with observations into a probabilistic program.
We provide proof rules that yield the exact expected runtime of this program in a fully automated fashion.
We implemented our approach and successfully analyzed various real--world BNs taken from the Bayesian network repository.

\end{abstract}
%


\section{Introduction}

\paragraph{Bayesian networks} (BNs) are \emph{probabilistic graphical models} representing joint probability distributions of sets of random variables with conditional dependencies.
Graphical models are a popular and appealing modeling formalism, as they allow to succinctly represent complex distributions in a human--readable way.
Bayesian Networks have been intensively studied at least since 1985~\cite{pearl1985bayesian} and
have a wide range of applications including
machine learning~\cite{DBLP:series/sci/Heckerman08}, 
speech recognition~\cite{DBLP:conf/aaai/ZweigR98}, 
sports betting~\cite{DBLP:journals/kbs/ConstantinouFN12}, 
gene regulatory networks~\cite{DBLP:conf/recomb/FriedmanLNP00}, 
diagnosis of \mbox{diseases~\cite{DBLP:journals/jamia/JiangC10}, 
and finance \cite{neapolitan2010probabilistic}}.

\paragraph{Probabilistic programs} are programs with the key ability to draw values at random.
Seminal papers by Kozen from the 1980s consider formal semantics~\cite{DBLP:journals/jcss/Kozen81} as well as initial work on verification~\cite{DBLP:journals/siamcomp/SharirPH84,DBLP:journals/jcss/Kozen85}.
McIver \& Morgan~\cite{mciver} build on this work to further weakest--precondition style verification for imperative probabilistic programs.

The interest in probabilistic programs has been rapidly growing in recent years~\cite{DBLP:conf/popl/Goodman13,DBLP:conf/icse/GordonHNR14}.
Part of the reason for this d\'{e}j\`{a} vu is 
their use for representing probabilistic graphical models~\cite{DBLP:books/daglib/0023091} such as BNs.
The full potential of modern probabilistic programming languages like 
\textsf{Anglican}~\cite{DBLP:conf/aistats/WoodMM14}, 
\textsf{Church}~\cite{DBLP:conf/uai/GoodmanMRBT08},
\textsf{Figaro}~\cite{pfeffer2009figaro},
\textsf{R2}~\cite{DBLP:conf/aaai/NoriHRS14}, or
\textsf{Tabular}~\cite{DBLP:conf/popl/GordonGRRBG14} 
is that they enable rapid prototyping and obviate the need to manually provide inference methods tailored to an individual model.

\paragraph{Probabilistic inference}
is the problem of determining the probability of an event given observed evidence.
It is a major problem for both BNs and probabilistic programs, and has been subject to intense investigations by both theoreticians and practitioners for more than three decades; see~\cite{DBLP:books/daglib/0023091} for a survey.
In particular, it has been shown that for probabilistic programs exact inference is highly undecidable~\cite{DBLP:conf/mfcs/KaminskiK15}, while for BNs both \emph{exact inference} as well as \emph{approximate inference} to an arbitrary precision are \textsc{NP}--hard~\cite{DBLP:journals/ai/Cooper90,DBLP:journals/ai/DagumL93}.
In light of these complexity--theoretical hurdles, a popular way to analyze probabilistic graphical models as well as probabilistic programs
is to gather a large number of independent and identically distributed (i.i.d.\ for short) samples and then do statistical reasoning on these samples.
In fact, all of the aforementioned probabilistic programming languages support sampling based inference methods.

\paragraph{Rejection sampling}
is a fundamental approach to obtain valid samples from BNs with observed evidence.
In a nutshell, this method first samples from the joint (unconditional) distribution of the BN. If the sample complies with all evidence, it is valid and accepted; otherwise it is rejected and one has to resample.

Apart from rejection sampling, there are more sophisticated sampling techniques, which mainly fall in two categories: Markov Chain Monte Carlo (MCMC) and importance sampling.
But while MCMC requires heavy hand--tuning and suffers from slow convergence rates on real--world instances~\cite[Chapter 12.3]{DBLP:books/daglib/0023091},
virtually all variants of importance sampling rely again on rejection sampling \cite{DBLP:books/daglib/0023091,DBLP:journals/mcm/YuanD06}.
%

A major problem with rejection sampling is that for poorly conditioned data, this approach might have to reject and resample very often in order to obtain just a single accepting sample.
Even worse, being poorly conditioned need not be immediately evident for a given BN, let alone a probabilistic program.
In fact, Gordon~et~al.~\cite[p.~177]{DBLP:conf/icse/GordonHNR14} point out that
\begin{quote}
``the main challenge in this setting [i.e. sampling based approaches] is that many samples that are generated during execution are ultimately rejected for not satisfying the observations.''
\end{quote}
If too many samples are rejected, the expected sampling time grows so large that sampling becomes infeasible.
The expected sampling time of a BN is therefore a key figure for deciding whether sampling based inference is the method of choice.

\paragraph{How long, O Bayesian network, will I sample thee?}

More precisely, we use techniques from program verification to give an answer to the following question:
\begin{quote}
Given a Bayesian network with observed evidence, how long does it take in expectation to obtain a \emph{single} sample that satisfies the observations?
\end{quote}
%
%
\begin{figure}[t]
\begin{tikzpicture}
\node (s) at (-2,0) [ellipse, draw, fill=gray!20, minimum width=2.2cm, minimum height=1.0cm
] {\scriptsize $\nodeS$};
\node (r) at (2,0) [ellipse, draw, fill=gray!20, minimum width=2.2cm, minimum height=1.0cm
] {\scriptsize $\nodeR$};
\node (g) at (0,-2) [ellipse, draw, fill=gray!20, minimum width=2.2cm, minimum height=1.0cm
] {\scriptsize $\nodeG$};

\draw[line width=1.0pt] (r) edge[->] (s) (r)edge[->](g) (g) (s)edge[->](g);

\node at (4.6, 0.2) {\begin{tabular}{|c | c|} \hline
\rowcolor{gray!20}\scriptsize $\nodeR = 0$ & \scriptsize $\nodeR = 1$  \\ \hline
\scriptsize $a$ &\scriptsize  $1-a$ \\
\hline
 \end{tabular}};
 
\node at (-5,0) {\begin{tabular}{ | c | c | c | } \hline
\rowcolor{gray!20} \phantom{\scriptsize ja}&\scriptsize $\nodeS = 0$&\scriptsize $\nodeS = 1$ \\ \hline
\scriptsize $\nodeR = 0$ & \scriptsize $a$ & \scriptsize $1-a$ \\
\hline
\scriptsize $\nodeR = 1$ & \scriptsize $0.2$ & \scriptsize $0.8$ \\
\hline
\end{tabular}};

  \node at (3.705,-2.8) {\begin{tabular}{ | c | c | c | } \hline
\rowcolor{gray!20} \phantom{\scriptsize ja}&\scriptsize $G=0$&\scriptsize $G=1$ \\ \hline
\scriptsize $\nodeS = 0$, $\nodeR =0$ & \scriptsize $0.01$ & \scriptsize $0.99$ \\
\hline
\scriptsize $\nodeS =0$, $\nodeR =1$ & \scriptsize $0.25$ & \scriptsize $0.75$ \\
\hline
\scriptsize $\nodeS =1$, $\nodeR =0$ & \scriptsize $0.9$ & \scriptsize $0.1$ \\
\hline
\scriptsize $\nodeS =1$, $\nodeR =1$ & \scriptsize $0.2$ & \scriptsize $0.8$ \\
\hline
\end{tabular}};

\end{tikzpicture}
\caption{
A simple Bayesian network.
}
\label{fig:intro:bn}
\end{figure}
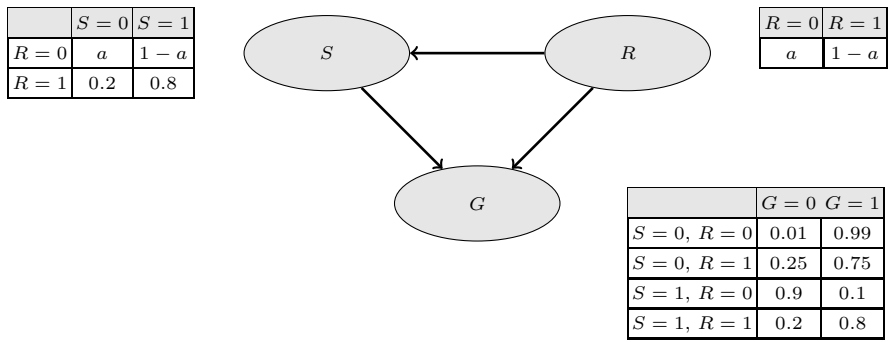
As an example, consider the BN in Figure~\ref{fig:intro:bn} which consists of just three nodes (random variables) that can each assume values $0$ or $1$.
Each node $X$ comes with a conditional probability table determining the probability of $X$ assuming some value given the values of all nodes $Y$ that $X$ depends on (i.e. $X$ has an incoming edge from $Y$), see Appendix \ref{app:example:calculations} for detailed calculations.
For instance, the probability that $G$ assumes value $0$, given that $S$ and $R$ are both assume $1$, is $0.2$.
Note that this BN is paramterized by $a \in [0,1]$.

Now, assume that our observed evidence is the event $G{=}0$ and we apply rejection sampling to obtain \emph{one} accepting sample from this BN. 
Then our approach will yield that a rejection sampling algorithm will, on average, require
\begin{align*}
	\frac{200 a^2 - 40 a - 460}{89 a^2 - 69 a - 21}
\end{align*}
guard evaluations, random assignments, etc. until it obtains a single sample that complies with the observation $G{=}0$ (the underlying runtime model is discussed in detail in Section \ref{sec:ert}).
By examination of this function, we see that for large ranges of values of $a$ the BN is rather well--behaved:
For $a \in [0.08,\, 0.78]$ the expected sampling time stays below 18. 
Above $a = 0.95$ the expected sampling time starts to grow rapidly up to $300$.

While 300 is still moderate, we will see later that expected sampling times of real--world BNs can be much larger.
For some BNs, the expected sampling time even exceeded $10^{18}$, rendering sampling based methods infeasible.
In this case, exact inference (despite \textsf{NP}--hardness) was a viable alternative (see Section~\ref{sec:implementation}).

\paragraph{Our approach.}
We apply weakest precondition style reasoning a l\'a McIver \& Morgan~\cite{mciver} and Kaminski et al.~\cite{DBLP:conf/esop/KaminskiKMO16} to analyze both expected outcomes and \emph{expected runtimes} (ERT) of a \emph{syntactic fragment of \pgcl}, which we call  the \emph{Bayesian Network Language} (\bnl).
Note that since \bnl is a syntactic fragment of \pgcl, every \bnl program is a \pgcl program but \emph{not vice versa}.
The main restriction of \bnl is that (in contrast to \pgcl) loops are of a special form that prohibits undesired data flow across multiple loop iterations.
While this restriction renders \bnl incapable of, for instance, counting the number of loop iterations\footnote{An example of a program that is \emph{not} expressible in \bnl is given in Example~\ref{ex:geo}.}, \bnl is expressive enough to encode Bayesian networks \mbox{with observed evidence}.

For \bnl, we develop dedicated proof rules to determine \emph{exact} expected values and the \emph{exact} ERT of any \bnl program, including loops,
without any user--supplied data, such as invariants~\cite{mciver,DBLP:conf/esop/KaminskiKMO16}, ranking or metering functions~\cite{DBLP:conf/cade/FrohnNHBG16}, (super)martingales~\cite{DBLP:conf/cav/ChakarovS13,DBLP:conf/popl/ChatterjeeFNH16,DBLP:conf/popl/ChatterjeeNZ17}, etc. 

As a central notion behind these rules, we introduce \emph{$f$--i.i.d.--ness} of probabilistic loops, a concept closely related to stochastic independence, that allows us to \emph{rule out undesired parts of the data flow across loop iterations}.
Furthermore, we show how every BN with observations is translated into a \bnl \mbox{program, such that}
\begin{enumerate}
  \item[(a)] executing the \bnl program corresponds to sampling from the \emph{conditional} joint distribution given by the BN and observed data, and
  \item[(b)] the ERT of the \bnl program corresponds to the expected time until a sample that satisfies the observations is obtained from the BN. 
\end{enumerate}
As a consequence, exact expected sampling times of BNs can be inferred by means of weakest precondition reasoning in a fully automated fashion.
This can be seen as a first step towards formally evaluating the quality of a plethora of different sampling methods (cf.~\cite{DBLP:books/daglib/0023091,DBLP:journals/mcm/YuanD06}) on source code level.

\paragraph{Contributions.} To summarize, our main contributions are as follows:
\begin{itemize}
  \item We develop easy--to--apply proof rules to reason about expected outcomes and expected runtimes of probabilistic programs with $f$--i.i.d. loops.
  \item We study a syntactic fragment of probabilistic programs, the Bayesian network language (\bnl), and show that our proof rules are applicable to every \bnl program; expected runtimes of $\bnl$ programs can thus be inferred.
  \item We give a formal translation from Bayesian networks with observations to $\bnl$ programs; expected sampling times of BNs can thus be inferred.
  \item We implemented a prototype tool that automatically analyzes the expected sampling time of BNs with observations. 
          An experimental evaluation on real--world BNs demonstrates that very large expected sampling times (in the magnitude of millions of years) can be inferred within less than a second; 
          This provides practitioners the means to decide whether sampling based methods are appropriate for their models. 
\end{itemize}

\paragraph{Outline.}
We discuss related work in Section~\ref{sec:related-work}.
Syntax and semantics of the probabilistic programming language $\pgcl$ are presented in Section~\ref{sec:pprogs}.
Our proof rules are introduced in Section~\ref{sec:ert-rules} and applied to BNs in Section~\ref{sec:applications}.
Section~\ref{sec:implementation} reports on experimental results and Section~\ref{sec:conclusion} concludes.


\section{Related Work}\label{sec:related-work}

While various techniques for formal reasoning about runtimes and expected outcomes of probabilistic programs have been developed, 
e.g.~\cite{Hehner:FAC:2011,McIver:FM:2005,DBLP:journals/jcss/BrazdilKKV15,luis,DBLP:conf/sas/Monniaux01}, none of them explicitly apply formal methods to reason about Bayesian networks on source code level.
In the following, we focus on approaches close to our work.

\paragraph{Weakest preexpectation calculus.}
Our approach builds upon the expected runtime calculus~\cite{DBLP:conf/esop/KaminskiKMO16}, which is itself based on work by Kozen~\cite{DBLP:journals/jcss/Kozen81,DBLP:journals/jcss/Kozen85} and McIver and Morgan~\cite{mciver}.
In contrast to~\cite{DBLP:conf/esop/KaminskiKMO16}, we develop specialized proof rules for a clearly specified program fragment \emph{without} requiring user--supplied invariants.
Since finding invariants often requires heavy calculations, our proof rules contribute towards simplifying and automating verification of probabilistic programs.

\paragraph{Ranking supermartingales.}
Reasoning about almost--sure termination is often based on ranking (super)martingales (cf.~\cite{DBLP:conf/cav/ChakarovS13,DBLP:conf/popl/ChatterjeeNZ17}).
In particular, Chatterjee et al.~\cite{DBLP:conf/popl/ChatterjeeFNH16} consider the class of affine probabilistic programs for which linear ranking supermartingales exist (\textsc{Lrapp}); thus proving (positive\footnote{Positive almost--sure termination means termination in finite expected time~\cite{DBLP:conf/rta/BournezG05}.}) almost--sure termination for all programs within this class.
They also present a doubly--exponential algorithm to approximate ERTs of \textsc{Lrapp} programs.
While all \bnl programs lie within \textsc{Lrapp}, our proof rules yield \emph{exact} ERTs as \emph{expectations} (thus allowing for compositional proofs), in contrast to a single number for a fixed initial state. 

\paragraph{Bayesian networks and probabilistic programs.}
Bayesian networks are a --- if not the most --- popular probabilistic graphical model (cf.~\cite{bishop,DBLP:books/daglib/0023091} for details) for reasoning about conditional probabilities.
They are closely tied to (a fragment of) probabilistic programs.
For example, \textsc{Infer.NET}~\cite{infernet} performs inference by compiling a probabilistic program into a Bayesian network.
While correspondences between probabilistic graphical models, such as BNs, have been considered in the literature~\cite{DBLP:conf/nips/MinkaW08,DBLP:conf/uai/GoodmanMRBT08,DBLP:conf/icse/GordonHNR14},
we are not aware of a formal soudness proof for a translation from classical BNs into probabilistic programs including conditioning.

Conversely, some probabilistic programming languages such as \textsc{Church}~\cite{DBLP:conf/uai/GoodmanMRBT08}, \textsc{Stan}~\cite{hoffman2014no}, and 
\textsc{R2}~\cite{DBLP:conf/aaai/NoriHRS14}
directly perform inference on the program level using sampling techniques similar to those developed for Bayesian networks.
Our approach is a step towards understanding sampling based approaches formally: 
We obtain the exact expected runtime required to generate a sample that satisfies all observations.
This may ultimately be used to evaluate the quality of a plethora of proposed sampling methods for Bayesian inference (cf.~\cite{DBLP:books/daglib/0023091,DBLP:journals/mcm/YuanD06}).


\section{Probabilistic Programs}\label{sec:pprogs}

We briefly present the probabilistic programming language that is used throughout this paper.
Since our approach is embedded into weakest-precondition style approaches, we also recap calculi for reasoning about both expected outcomes and expected runtimes of probabilistic programs.

\subsection{The Probabilistic Guarded Command Language}

We enhance Dijkstra's Guarded Command Language \cite{DBLP:books/ph/Dijkstra76,DBLP:journals/cacm/Dijkstra75} by a probabilistic construct, namely a random assignment.
We thereby obtain a \emph{probabilistic Guarded Command Language} (for a closely related language, see~\cite{mciver}).
%

%
	Let $\Vars$ be a finite set of \emph{program variables}.
    Moreover, let $\Rats$ be the set of rational numbers, and let $\Dists{\Rats}$ be the set of discrete probability distributions over $\Rats$.
	The set of \emph{program states} is given by
		$\States = \setcomp{\sigma}{\sigma \colon \Vars \To \Rats}$.
	
	A \emph{distribution expression} $\mu$ is a function of type
	%
		$\mu \colon \States \To \Dists{\Rats}$
	%
	that takes a program state and maps it to a probability distribution on values from $\Rats$.
	We denote by $\mu_\sigma$ the distribution obtained from applying $\sigma$ to $\mu$.
	
	The probabilistic guarded command language ($\pgcl$) is given by the grammar
	\begin{align*}
		C ~~\longrightarrow~~ &\quad\SKIP \tag{effectless program} \\
		& \quad |~~ \DIVERGE \tag{endless loop} \\
		& \quad |~~ \PASSIGN{x}{\mu} \tag{random assignment} \\
		& \quad |~~ \COMPOSE{C}{C} \tag{sequential composition} \\
		& \quad |~~ \ITE{\varphi}{C}{C}\quad{} & \tag{conditional choice} \\
		& \quad |~~ \WHILEDO{\varphi}{C} & \tag{while loop} \\
		& \quad |~~ \REPEATUNTIL{C}{\varphi}~, & \tag{repeat--until loop}
	\end{align*}
	where $x \in \Vars$ is a program variable, $\mu$ is a distribution expression, and $\varphi$ is a Boolean expression guarding a choice or a loop.
    A $\pgcl$ program that contains neither $\DIVERGE$, nor $\mathtt{while}$, nor $\mathtt{repeat-until}$ loops is called loop--free.
	
	For $\sigma \in \States$ and an arithmetical expression $E$ over $\Vars$, we denote by $\sigma(E)$ the evaluation of $E$ in $\sigma$, i.e.\ the value that is obtained by evaluating $E$ after replacing any occurrence of any program variable $x$ in $E$ by the value $\sigma(x)$.
	Analogously, we denote by $\sigma(\varphi)$ the evaluation of a guard $\varphi$ in state $\sigma$ to either $\true$ or $\false$.
	Furthermore, for a value $v \in \Rats$ we write $\sigma\statesubst{x}{v}$ to indicate that we set program variable $x$ to value $v$ in program state $\sigma$, i.e.\footnote{We use $\lambda$--expressions to construct functions: Function $\lambda X \mydot \epsilon$ applied to an argument $\alpha$ evaluates to $\epsilon$ in which every occurrence of $X$ is replaced by $\alpha$.}
	\begin{align*}
		\sigma\statesubst{x}{v} ~=~ \lambda\, y\mydot \begin{cases}
			v, & \textnormal{if } y = x\\
			\sigma(y), & \textnormal{if } y \neq x~.
		\end{cases}
	\end{align*}
	
	We use the Iverson bracket notation
	to associate with each guard its according indicator function.
	Formally, the Iverson bracket $\iverson{\varphi}$ of $\varphi$ is thus defined as the function
		$\iverson{\varphi} = \lambda\, \sigma \mydot \sigma(\varphi)$.
%
%

Let us briefly go over the $\pgcl$ constructs and their effects:
$\SKIP$ does not alter the current program state.
The program $\DIVERGE$ is an infinite busy loop, thus takes infinite time to execute.
It returns no final state whatsoever.

The random assignment $\PASSIGN{x}{\mu}$ is (a) the only construct that can actually alter the program state and (b) the only construct that may introduce random behavior into the computation. It takes the current program state $\sigma$, then \emph{samples} a value $v$ from probability distribution $\mu_\sigma$, and then assigns $v$ to program variable $x$.
An example of a random assignment is
\begin{align*}
	\PASSIGN{x}{\sfrac{1}{2} \cdot \langle 5 \rangle + \sfrac{1}{6} \cdot \langle y + 1 \rangle + \sfrac{1}{3} \cdot \langle y - 1 \rangle}~.
\end{align*}
If the current program state is $\sigma$, then the program state is altered to either $\sigma\statesubst{x}{5}$ with probability $\sfrac{1}{2}$, or to $\sigma\statesubst{x}{\sigma(y) + 1}$ with probability $\sfrac{1}{6}$, or to $\sigma\statesubst{x}{\sigma(y) - 1}$ with probability $\sfrac{1}{3}$.
The remainder of the \pgcl constructs are standard programming language constructs.

%
%
%

In general, a $\pgcl$ program $C$ is executed on an input state and yields a \emph{probability distribution} over final states due to possibly occurring random assignments inside of $C$.
We denote that resulting distribution by $\semantics{C}{\sigma}$.
Strictly speaking, programs can yield \emph{subdistributions}, i.e.\ probability distributions whose total mass may be below 1.
The ``missing" probability mass represents the probability of nontermination.
Let us conclude our presentation of \pgcl \mbox{with an example}:
\begin{example}[Geometric Loop]
\label{ex:geo}
	Consider the program $C_\mathit{geo}$ given by
	\begin{align*}
		&\COMPOSE{\PASSIGN{x}{0}}{\quad\PASSIGN{c}{\sfrac{1}{2} \cdot \langle 0 \rangle + \sfrac{1}{2} \cdot \langle 1 \rangle}}; \\
		&\WHILE{c = 1} \COMPOSE{\PASSIGN{x}{x + 1}}{\PASSIGN{c}{\sfrac{1}{2} \cdot \langle 0 \rangle + \sfrac{1}{2} \cdot \langle 1 \rangle}}\}
	\end{align*}
	This program basically keeps flipping coins until it flips, say, heads ($c=0$).
	In $x$ it counts the number of unsuccessful trials.\footnote{This counting is also the reason that $C_{\mathit{geo}}$ is an example of a program that is not expressible in our \bnl language that we present later.}
	In effect, it almost surely sets $c$ to $0$ and moreover it establishes a geometric distribution on $x$.
	The resulting distribution is given by
	\begin{align*}
		\semantics{C_\mathit{geo}}{\sigma}(\tau) \eeq \sum_{n = 0}^{\omega} \iverson{\tau = \sigma\statesubst{c,x}{0,n}} \cdot \frac{1}{2^{n+1}}~.\tag*{$\triangle$}
	\end{align*}
\end{example}

\subsection{The Weakest Preexpectation Transformer}

We now present the weakest preexpectation transformer $\wpsymbol$  for reasoning about expected outcomes of executing probabilistic programs in the style of McIver \& Morgan~\cite{mciver}.
Given a random variable $f$ mapping program states to reals, it allows us to reason about the expected value of $f$ after executing a probabilistic program on a given state.

\subsubsection{Expectations.}
The random variables the $\wpsymbol$ transformer acts upon are taken from a set of so-called expectations, a term coined by McIver \& Morgan~\cite{mciver}:
\begin{definition}[Expectations]
\label{def:expectations}
	The set of expectations $\E$ is defined as
	\begin{align*}
		\E ~{}={}~ \left\{f ~\middle|~ f\colon \Sigma \To \Rposinf\right\} ~.
	\end{align*}
	
	We will use the notation $f\subst{x}{E}$ to indicate the \emph{replacement} of every occurrence of $x$ in $f$ by $E$.
	Since $x$, however, does not actually occur in $f$, we more formally define $f\subst{x}{E} = \lambda \sigma\mydot f(\sigma\statesubst{x}{\sigma(E)})$.

	A complete partial order $\leq$ on $\E$ is obtained by point--wise lifting the canonical total order on $\Rposinf$, i.e.\
	\begin{align*}
		f_1 \ppreceq f_2 \quad\text{iff}\quad \forall \sigma\in\Sigma\colon~~ f_1(\sigma) \lleq f_2(\sigma) ~.
	\end{align*}
	Its least element is given by $\lambda \sigma\mydot 0$ which we (by slight abuse of notation) also denote by $0$.
    Suprema are constructed pointwise, i.e.\ for $S \subseteq \E$ the supremum $\sup S$ is given by
		$\sup S = \lambda \sigma\mydot \sup_{f \in S}~f(\sigma)$. 
\end{definition}
\noindent
We allow expectations to map only to positive reals, so that we have a complete partial order readily available, which would not be the case for expectations of type $\States \To \Reals \cup \{-\infty,\, +\infty\}$.
A $\wpsymbol$ calculus that \emph{can} handle expectations of such type needs more technical machinery and cannot make use of this underlying natural partial order~\cite{lics17}.
Since we want to reason about ERTs which are by nature non--negative, we will not need such complicated calculi.

Notice that we use a slightly different definition of expectations than McIver \& Morgan~\cite{mciver}, as we allow for \emph{unbounded} expectations, whereas \cite{mciver} requires that expectations are \emph{bounded}.
This however would prevent us from capturing ERTs, which are potentially unbounded.

\subsubsection{Expectation Transformers.}

For reasoning about the expected value of $f \in \E$ after execution of $C$, we employ a backward--moving weakest preexpectation transformer
%
	$\wpsymbol \llbracket C \rrbracket \colon \E \To \E$,
%
that maps a \emph{postexpectation} $f \in \E$ to a \emph{preexpectation} $\wp{C}{f}  \in \E$, such that $\wp{C}{f}(\sigma)$ is the expected value of $f$ after executing $C$ on initial state $\sigma$. 
Formally, if $C$ executed on input $\sigma$ yields final distribution $\semantics{C}{\sigma}$, then
the \emph{weakest preexpectation} $\wp{C}{f}$ \emph{of} $C$ \emph{with respect to postexpectation} $f$ \emph{is given by}
\abovedisplayskip=0pt
\begin{align}
	\wp{C}{f}(\sigma) \eeq \Exp{\Sigma}{f}{\semantics{C}{\sigma}}~,\label{eq:defwp}
\end{align}
\normalsize
where we denote by $\Exp{A}{h}{\nu}$ the expected value of a random variable $h\colon A \To \Rposinf$ with respect to a probability distribution $\nu \colon A \To [0,\, 1]$.
Weakest preexpectations
can be defined in a very systematic way:
\begin{definition}[The $\textbf{\textsf{wp}}$ Transformer~\textnormal{\cite{mciver}}]
\label{def:wp}
	The weakest preexpectation transformer $\wpsymbol\colon \pgcl \To \E \To \E$ is defined by induction on all $\pgcl$ programs according to the rules in \autoref{table:wp}.
	We call $F_f(X) = \iverson{\neg \varphi} \cdot f + \iverson{\varphi} \cdot \wp{C}{X}$ the \emph{$\wpsymbol$--characteristic functional} of the loop $\WHILEDO{\varphi}{C}$ with respect to postexpectation $f$.
For a given $\wpsymbol$--characteristic function $F_f$, we call the sequence $\{F_f^n(0) \}_{n\in\Nats}$ the \emph{orbit of $F_f$}.
\end{definition}
\begin{table}[t]
\centering
\renewcommand{\arraystretch}{1.5}
\begin{tabular}{@{\hspace{1em}}l@{\hspace{2em}}l@{\hspace{1em}}}
	\hline\hline
	$\boldsymbol{C}$			& $\boldsymbol{\textbf{\textsf{wp}}\,\left \llbracket C\right\rrbracket  \left(f \right)}$\\
	\hline\hline
	$\SKIP$					& $f$ \\
	$\DIVERGE$				& $0$ \\
	$\PASSIGN{x}{\mu}$			& $\lambda \sigma\mydot \Exp{\Rats}{\big(\lambda v\mydot f\subst{x}{v}\big)}{\mu_\sigma}$ \\
	$\ITE{\varphi}{C_1}{C_2}$		& $\iverson{\varphi} \cdot \wp{C_1}{f} + \iverson{\neg \varphi} \cdot \wp{C_2}{f}$ \\
	$\COMPOSE{C_1}{C_2}$		& $\wp{C_1}{\wp{C_2}{f}}$ \\
	$\WHILEDO{\varphi}{C'}$		& $\lfp X\mydot \iverson{\neg \varphi} \cdot f + \iverson{\varphi} \cdot \wp{C'}{X}$\\
	$\REPEATUNTIL{C'}{\varphi}$	& $\wp{\COMPOSE{C'}{\WHILEDO{\neg \varphi}{C'}}}{f}$\\
	\hline
\end{tabular}
\vspace{1ex}
\caption{Rules for the $\wpsymbol$--transformer.}
\label{table:wp}
\end{table}
\noindent
Let us briefly go over the definitions in \autoref{table:wp}:
For $\SKIP$ 
the program state is not altered and thus the expected value of $f$ is just $f$.
The program $\DIVERGE$ will never yield any final state. 
The distribution over the final states yielded by $\DIVERGE$ is thus the null distribution $\nu_0(\tau) = 0$, that assigns probability 0 to \emph{every} state.
Consequently, the expected value of $f$ after execution of $\DIVERGE$ is given by $\Exp{\Sigma}{f}{\nu_0} = \sum_{\tau \in \Sigma}0 \cdot f(\tau) = 0$.

The rule for the random assignment $\PASSIGN{x}{\mu}$ is a bit more technical:
Let the current program state be $\sigma$.
Then for every value $v \in \Rats$, the random assignment assigns $v$ to $x$ with probability $\mu_\sigma(v)$, where $\sigma$ is the current program state.
The value of $f$ after assigning $v$ to $x$ is $f(\sigma\statesubst{x}{v}) = f\subst{x}{v}(\sigma)$ and therefore the expected value of $f$ after executing the random assignment is given by
\begin{align*}
	\sum_{v \in \Rats} \mu_\sigma(v) \cdot f\subst{x}{v}(\sigma) 
	\eeq \Exp{\Rats}{\big(\lambda v \mydot f\subst{x}{v}(\sigma) \big)}{\mu_\sigma}~.
\end{align*}
Expressed as a function of $\sigma$, the latter yields precisely the definition in \autoref{table:wp}.

The definition for the conditional choice $\ITE{\varphi}{C_1}{C_2}$ is not surprising:
if the current state satisfies $\varphi$, we have to opt for the weakest preexpectation of $C_1$, whereas if it does not satisfy $\varphi$, we have to choose the weakest preexpectation of $C_2$.
This yields precisely the definition in \autoref{table:wp}.

The definition for the sequential composition $\COMPOSE{C_1}{C_2}$ is also straightforward:
We first determine $\wp{C_2}{f}$ to obtain the expected value of $f$ after executing $C_2$.
Then we mentally prepend the program $C_2$ by $C_1$ and therefore determine the expected value of $\wp{C_2}{f}$ after executing $C_1$.
This gives the weakest preexpectation of $\COMPOSE{C_1}{C_2}$ with respect to postexpectation $f$.

The definition for the while loop makes use of a least fixed point, which is a standard construction in program semantics.
Intuitively, the fixed point iteration of the $\wpsymbol$--characteristic functional, given by $0,\, F_f(0),\, F_f^2(0),\, F_f^3(0),\, \ldots$, corresponds to the portion the expected value of $f$ after termination of the loop, that can be collected within at most $0,\, 1,\, 2,\, 3,\, \ldots$ loop guard evaluations.
The Kleene Fixed Point Theorem~\cite{DBLP:journals/ipl/LassezNS82} ensures that this iteration converges to the least fixed point, i.e.\
\begin{align*}
	\sup_{n \in \Nats} F_f^n(0) \eeq \lfp F_f \eeq \wp{\WHILEDO{\varphi}{C}}{f}~.
\end{align*}
By inspection of the above equality, we see that the least fixed point is exactly the construct that we want for while loops, since $\sup_{n \in \Nats} F_f^n(0)$ in principle allows the loop to run for any number of iterations, which captures precisely the semantics of a while loop, where the number of loop iterations is --- in contrast to e.g.\ \texttt{for} loops --- not determined upfront.

Finally, since $\REPEATUNTIL{C}{\varphi}$ is syntactic sugar for $\COMPOSE{C}{\WHILEDO{\varphi}{C}}$, we simply define the weakest preexpectation of the former as the weakest preexpectation of the latter.
Let us conclude our study of the effects of the $\wpsymbol$ transformer by means of an example: 
%
%
\begin{example}
\label{ex:wp-loop-free}
Consider the following program $C$:
\begin{align*}
	&\PASSIGN{c}{\sfrac{1}{3} \cdot \langle 0 \rangle + \sfrac{2}{3} \cdot \langle 1 \rangle}; \\
	&\ITE{c = 0}{\PASSIGN{x}{\sfrac{1}{2} \cdot \langle 5 \rangle + \sfrac{1}{6} \cdot \langle y + 1 \rangle + \sfrac{1}{3} \cdot \langle y - 1 \rangle}}{\SKIP}
\end{align*}
Say we wish to reason about the expected value of $x + c$ after execution of the above program.
We can do so by calculating $\wp{C}{x + c}$ using the rules in Table~\ref{table:wp}.
This calculation in the end yields
	$\wp{C}{x + c} \eeq \sfrac{3y + 26}{18}$
The expected valuation of the expression $x + c$ after executing $C$ is thus $\sfrac{3y + 26}{18}$.
Note that $x + c$ can be thought of as an expression that is evaluated in the final states after execution, whereas $\sfrac{3y + 26}{18}$ must be evaluated in the initial state before execution of $C$.
\hfill$\triangle$
\end{example}
\noindent

\subsubsection{Healthiness Conditions of \textsf{\textbf{wp}}.}
The $\wpsymbol$ transformer enjoys some useful properties, sometimes called \emph{healthiness conditions}~\cite{mciver}.
Two of these healthiness conditions that we will heavily make use of are given below:
\begin{theorem}[Healthiness Conditions for the $\textbf{\textsf{wp}}$ Transformer~\textnormal{\cite{mciver}}]
	\label{thm:basic-prop}
	For all $C \in \pgcl$, $f_1, f_2 \in \E$, and $a \in \Rpos$, the following holds:
	\begin{enumerate}
		\item \qquad \label{thm:basic-prop-linearity}
			$\wp{C}{a \cdot f_1 + f_2} \eeq a \cdot \wp{C}{f_1} + \wp{C}{f_2}$ \hfill(linearity)\\[-.75em]
		\item \qquad \label{thm:basic-prop-strictness}
			$\wp{C}{0} \eeq 0$ \hfill (strictness)
	\end{enumerate}
\end{theorem}

\subsection{The Expected Runtime Transformer}

\label{sec:ert}

While for deterministic programs we can speak of \emph{the} runtime of a program on a given input, the situation is different for probabilistic programs:
For those we instead have to speak of the \emph{expected runtime} (ERT).
Notice that the ERT can be finite (even constant) while the program may still admit infinite executions.
An example of this is the geometric loop in Example~\ref{ex:geo}.

A $\wpsymbol$--like transformer designed specifically for reasoning about ERTs is the $\ertsymbol$ transformer~\cite{DBLP:conf/esop/KaminskiKMO16}.
Like $\wpsymbol$, it is of type $\ertsymbol\llbracket C \rrbracket \colon \E \To \E$ and it can be shown that
%
	$\ert{C}{0}(\sigma)$
%
is precisely the \emph{expected runtime of executing $C$ on input $\sigma$}.
More generally, if $f\colon \Sigma \To \Rposinf$ measures the time that is needed after executing $C$ (thus $f$ is evaluated in the final states after termination of $C$), then $\ert{C}{f}(\sigma)$ is the expected time that is needed to run $C$ on input $\sigma$ and then let time $f$ pass.
For a more in--depth treatment of the $\ertsymbol$ transformer, see~\cite[Section 3]{DBLP:conf/esop/KaminskiKMO16}.
The transformer is defined as follows:
\begin{definition}[The $\boldertsymbol$ Transformer~\textnormal{\cite{DBLP:conf/esop/KaminskiKMO16}}]
\label{def:ert}
	The expected runtime transformer $\ertsymbol\colon \pgcl \To \E \To \E$ is defined by induction on all $\pgcl$ programs according to the rules given in \autoref{table:ert}.
	We call $F_f(X) = 1 + \iverson{\neg \varphi} \cdot f + \iverson{\varphi} \cdot \wp{C}{X}$ the \emph{$\ertsymbol$--characteristic functional} of the loop $\WHILEDO{\varphi}{C}$ with respect to postexpectation $f$.
As with $\wpsymbol$, for a given $\ertsymbol$--characteristic function $F_f$, we call the sequence $\{F_f^n(0) \}_{n\in\Nats}$ the \emph{orbit of $F_f$}.
Notice that 
\begin{align*}
	\ert{\WHILEDO{\varphi}{C}}{f} \eeq \lfp F_f \eeq \sup~\{F_f^n(0)\}_{n\in\Nats}~.
\end{align*}
\normalsize
\end{definition}
\begin{table}[t]
\centering
\renewcommand{\arraystretch}{1.5}
\begin{tabular}{@{\hspace{1em}}l@{\hspace{2em}}l@{\hspace{1em}}}
	\hline\hline
	$\boldsymbol{C}$			& $\boldsymbol{\textbf{\textsf{ert}}\,\left \llbracket C\right\rrbracket  \left(f \right)}$\\
	\hline\hline
	$\SKIP$					& $1 + f$ \\
	$\DIVERGE$				& $\infty$ \\
	$\PASSIGN{x}{\mu}$			& $1 + \lambda \sigma\mydot \Exp{\Rats}{\big(\lambda v\mydot f\subst{x}{v}\big)}{\mu_\sigma}$ \\
	$\ITE{\varphi}{C_1}{C_2}$		& $1 + \iverson{\varphi} \cdot \ert{C_1}{f} + \iverson{\neg \varphi} \cdot \ert{C_2}{f}$ \\
	$\COMPOSE{C_1}{C_2}$		& $\ert{C_1}{\big(\ert{C_2}{f}\big)}$ \\
	$\WHILEDO{\varphi}{C'}$			& $\lfp X\mydot 1 + \iverson{\neg \varphi} \cdot f + \iverson{\varphi} \cdot \ert{C'}{X}$\\
	$\REPEATUNTIL{C'}{\varphi}$	& $\ert{\COMPOSE{C'}{\WHILEDO{\neg\varphi}{C'}}}{f}$\\
	\hline
\end{tabular}
\vspace{1ex}
\caption{Rules for the $\ertsymbol$--transformer.}
\label{table:ert}
\end{table}
\noindent
The rules for $\ertsymbol$ are very similar to the rules for $\wpsymbol$.
The runtime model we assume is that $\SKIP$ statements, random assignments, and guard evaluations for both conditional choice and while loops cost one unit of time.
This runtime model can easily be adopted to count only the number of loop iterations or only the number of random assignments, etc.
We conclude with a strong connection between the $\wpsymbol$ and the $\ertsymbol$ transformer, that is crucial in our proofs:
\begin{theorem}[Decomposition of $\boldertsymbol$ \textnormal{\cite{DBLP:conf/lics/OlmedoKKM16}}]
\label{thm:ert-wp}
       For any $C \in \pgcl$ and $f \in \E$,
       \begin{align*}
              \ert{C}{f} ~{}={}~ \ert{C}{0} + \wp{C}{f}~.
       \end{align*}
\end{theorem}
%
%



\section{Expected Runtimes of i.i.d.\ Loops}
\label{sec:ert-rules}

\newcommand{\kGuard}{\varphi} 
\newcommand{\kGuardB}{\iverson{\kGuard}}
\newcommand{\nkGuardB}{\iverson{\neg \kGuard}}
\newcommand{\exampleLoop}{C_{\mathit{flip}}}
\newcommand{\rejectionProg}{C_{\mathit{circle}}}
\newcommand{\rejectionProgBody}{C_{\mathit{body}}}
\noindent
We derive a proof rule that allows to determine \emph{exact ERTs of independent and identically distributed loops} (or \emph{i.i.d.\ loops} for short). 
Intuitively, a loop is i.i.d.\ if the distributions of states that are reached at the end of different loop iterations are equal.  
This is the case whenever there is no data flow across different iterations. 
In the non--probabilistic case, such loops either terminate after exactly one iteration or never. 
This is different for probabilistic programs.

As a running example, consider the program $\rejectionProg$ in Figure~\ref{fig:circle-loop}.
%
%
%
%
%
%
\begin{figure}[t]
\begin{center}
\begin{minipage}{0.40\textwidth}
       \begin{align*}
              & \WHILE{(x-5)^2 + (y-5)^2 \geq 25} \\
              & \qquad \PASSIGN{x}{\mathtt{Unif}[0\ldots 10]}; \\
              & \qquad \PASSIGN{y}{\mathtt{Unif}[0\ldots 10]} \\ 
              & \}
       \end{align*} 
\end{minipage}
\qquad
\begin{minipage}{0.4\textwidth}
\begin{center}
\scalebox{.75}{
\begin{tikzpicture}[scale=.3]

\draw[help lines, use as bounding box, white] (-1, -1) grid (11, 11);

\draw[->, thick] (-0.3, 0) -- (11, 0);
\draw[->, thick] (0, -0.3) -- (0, 11);

\foreach \x in {1, 2, ..., 10} {
	\draw (\x, -0.3) -- (\x, 0.3);
}
\node at (5, -0.7) {\scriptsize $5$};
\node at (10, -0.7) {\scriptsize $10$};
\foreach \y in {1, 2, ..., 10} {
	\draw (-0.3, \y) -- (0.3, \y);
}
\node at (-0.8, 5) {\scriptsize $5$};
\node at (-0.8, 10) {\scriptsize $10$};

\draw[thick, fill=shadegray] (0,0) -- (10,0) -- (10,10) -- (0,10) -- (0,0);
\draw[thick, fill=white] (5,5) circle (5cm);

\node at (0.9, 0.9) {\scriptsize $\times$};
\node at (0.7, 1.8) {\scriptsize $\times$};

\node at (6, 7) {\scriptsize $\times$};
\node at (2.3, 3.5) {\scriptsize $\times$};
\node at (7.1, 2.1) {\scriptsize $\times$};
\node at (2, 7.5) {\scriptsize $\times$};
\node at (6.4, 8.3) {\scriptsize $\times$};
\node at (2, 7.5) {\scriptsize $\times$};
\node at (8.4, 4.5) {\scriptsize $\times$};
\node at (5.1, 4.9) {\scriptsize $\times$};
\node at (8.9, 3.5) {\scriptsize $\times$};

\node at (9.5, 9.1) {\scriptsize $\times$};
\node at (9.2, 9.3) {\scriptsize $\times$};

\node at (0.5, 9.3) {\scriptsize $\times$};

\node at (9, 0.8) {\scriptsize $\times$};
\node at (2, 7.5) {\scriptsize $\times$};
\node at (8.2, 0.3) {\scriptsize $\times$};
\node at (9.2, 1.1) {\scriptsize $\times$};
\end{tikzpicture}
}
\end{center}
\end{minipage} 
\end{center}
\caption{An i.i.d.\ loop sampling a point within a circle uniformly at random using rejection sampling. 
The picture on the right--hand side visualizes the procedure: 
In each iteration a point ($\times$) is sampled.
If we obtain a point within the white area inside the square, we terminate. 
Otherwise, i.e.\ if we obtain a point within the gray area outside the circle, we resample.}
\label{fig:circle-loop}
\end{figure}
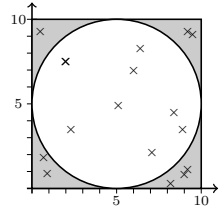
$\rejectionProg$ samples a point within a circle with center $(5,5)$ and radius $r=5$ uniformly at random using rejection sampling.
In each iteration, it samples a point $(x,y) \in [0, \ldots, 10]^2$ within the square (with some fixed precision). 
The loop ensures that we resample if a sample is not located within the circle.
Our proof rule will allow us to systematically determine the ERT of this loop, i.e. the average amount of time required until a single point within the circle is sampled.

Towards obtaining such a proof rule, we first present a syntactical notion of the i.i.d.\ property.
It relies on expectations that are not affected by a $\pgcl$ program:
%
%
\begin{definition}
  Let $C \in \pgcl$ and $f \in \E$.
  Moreover, let $\VarsAssign{C}$ denote the set of all variables 
  that occur on the left--hand side of an assignment in $C$, and let $\VarsInExp{f}$ be the set of all variables that ``occur in $f$", i.e.\ formally
  \begin{align*}
  	x \in \VarsInExp{f} \qqiff \exists\, \sigma~ \exists\, v, v' \colon \quad f(\sigma\statesubst{x}{v}) \nneq f(\sigma\statesubst{x}{v'})~.
  \end{align*}
  Then $f$ is \emph{unaffected} by $C$, denoted $\RelNewRule{f}{C}$, iff $\VarsInExp{f} \cap \VarsAssign{C} = \emptyset$. 
\end{definition}
\noindent
%
%
%
We are interested in expectations that are unaffected by $\pgcl$ programs because of a simple, yet useful observation:
If $\RelNewRule{g}{C}$, then \emph{$g$ can be treated like a constant} w.r.t.\  the transformer $\wpsymbol$ (i.e.\ like the $a$ in Theorem~\ref{thm:basic-prop}~(\ref{thm:basic-prop-linearity})).
For our running  example $\rejectionProg$ (see Figure~\ref{fig:circle-loop}), the expectation $f = \wp{\rejectionProgBody}{\iverson{x + y \leq 10}}$ is unaffected by the loop body $\rejectionProgBody$ of $\rejectionProg$. Consequently, we have $\wp{\rejectionProgBody}{f} = f \cdot \wp{\rejectionProgBody}{1} = f$.
In general, we obtain the following property:
%
%
%
\begin{lemma}[Scaling by Unaffected Expectations]
      \label{lem:newrule-main}
      Let $C\in \pgcl$ and $f,g \in \E$.
      Then $\RelNewRule{g}{C}$ implies $\wp{C}{g \cdot f} =  g \cdot \wp{C}{f}$.
\end{lemma}
\begin{proof}
By induction on the structure of $C$. See Appendix~\ref{proof-lem:newrule-main}. \qed
\end{proof}
\noindent
%
%
%
We develop a proof rule that only requires
that both the probability of the guard evaluating to true after one iteration of the loop body (i.e.\ $\wp{C}{\iverson{\kGuard}}$) as well as the expected value of $\iverson{\neg \varphi} \cdot f$ after one iteration (i.e.\ $\wp{C}{\iverson{\neg \kGuard} \cdot f}$) are unaffected by the loop body.
We thus define the following:
%
%
%
\begin{definition}[$\boldsymbol{f}$--Independent and Identically Distributed Loops]\label{def:f-iid}
       Let $C \in \pgcl$, $\kGuard$ be a guard, and $f \in \E$. 
       Then we call the loop $\WHILEDO{\kGuard}{C}$
       \emph{$f$--independent and identically distributed} (or \emph{$f$--i.i.d.} for short), if both
       \begin{align*}
              \RelNewRule{\wp{C}{\kGuardB}}{C} \qquad \text{and} \qquad
 \RelNewRule{\wp{C}{\nkGuardB \cdot f}}{C}~.
       \end{align*}
\end{definition}
\noindent
%
%
%
%
%
%
%
%
%
\begin{example}
        Our example program $\rejectionProg$ (see Figure~\ref{fig:circle-loop}) is $f$--i.i.d. for all $f \in \E$. 
        This is due to the fact that
       \begin{equation*}
              \wp{\rejectionProgBody}{\iverson{(x-5)^2 + (y-5)^2 \geq 25}}
              ~{}={}~ \RelNewRule{\frac{48}{121}}{\rejectionProgBody}
              \tag{by Table~\ref{table:wp}}
       \end{equation*}
       and (again for some fixed precision $p \in \Nats \setminus \{0\}$)
       \begin{align*}
             &\wp{\rejectionProgBody}{\iverson{(x-5)^2 + (y-5)^2 > 25} \cdot f} \\
             &{}={}~ \RelNewRule{\frac{1}{121} \cdot \sum\limits_{i=0}^{10p}\sum\limits_{j=0}^{10p} \iverson{(\nicefrac i p -5)^2 + (\nicefrac j p -5)^2 > 25} \cdot f[x/(\nicefrac i p),y/(\nicefrac j p)]}{\rejectionProgBody}~.
             \tag*{$\triangle$}
       \end{align*}
\end{example}
\noindent
Our main technical Lemma is that we can express the orbit of the $\wpsymbol$--characteristic function as a partial geometric series:
\begin{lemma}[Orbits of $\boldsymbol{f}$--i.i.d.\ Loops]\label{thm:newrule_wp_finite_approximations}
        Let $C \in \pgcl$,  $\kGuard$ be a guard, $f \in \E$ such that the loop  $\WHILEDO{\kGuard}{C}$ is $f$--i.i.d, and let $F_f$ be the corresponding $\wpsymbol$--characteristic function. 
        Then for all  $n \in \Nats \setminus \{ 0 \}$, it holds that 
         \begin{align*}
              &F_f^n(0)
              \eeq
              \iverson{\kGuard} \cdot \wp{C}{\iverson{\neg \kGuard} \cdot f}
               \cdot \left. \sum\limits_{i=0}^{n-2} \middle( \wp{C}{\iverson{\kGuard}}^i \right) ~{}+{}~ \iverson{\neg \kGuard} \cdot f~.
       \end{align*}
\end{lemma}
\begin{proof}
 By use of \autoref{lem:newrule-main}, see Appendix~\ref{lem:newrule_wp_finite_approximations}.
\end{proof}
\noindent
Using this precise description of the $\wpsymbol$ orbits, we now establish proof rules for $f$--i.i.d.\ loops, first for $\wpsymbol$ and later for $\ertsymbol$. 
\begin{theorem}[Weakest Preexpectations of $\boldsymbol{f}$--i.i.d.\ Loops] \label{thm:newrule-wpsemantics}
       Let $C \in \pgcl$, $\kGuard$ be a guard, and $f \in \E$.
       If the loop $\WHILEDO{\kGuard}{C}$ is $f$--i.i.d., then
                       %
                       %
                       \begin{align*}
                              \wp{\WHILEDO{\kGuard}{C}}{f} \eeq 
                              \iverson{\kGuard} \cdot 
                              \frac{\wp{C}{\iverson{\neg \kGuard} \cdot f}}{1-\wp{C}{\iverson{\kGuard}}}
                              + \iverson{\neg \kGuard} \cdot f ~,
                       \end{align*}
	where we define $\frac{0}{0} \coloneqq 0$.
\end{theorem}
%
%
\begin{proof}
	We have
        \begin{align*}
                 & \wp{\WHILEDO{\kGuard}{C}}{f} \\
        & \eeq  \sup_{n \in \Nats}~ F_f^n(0) \tag{by Definition~\ref{def:wp}}
                      \\
        & \eeq  \sup_{n \in \Nats}~
                      \iverson{\kGuard} \cdot  \wp{C}{\iverson{\neg \kGuard} \cdot f}
                     \cdot \left. \sum\limits_{i=0}^{n-2} \middle( \wp{C}{\iverson{\kGuard}}^i \right) 
                    + \iverson{\neg \kGuard} \cdot f 
                    \tag{by Lemma \ref{thm:newrule_wp_finite_approximations}} \\
        & \eeq \iverson{\kGuard} \cdot  \wp{C}{\iverson{\neg \kGuard} \cdot f}
                     \cdot \left. \sum_{i=0}^{\omega} \middle( \wp{C}{\iverson{\kGuard}}^i \right) 
                    + \iverson{\neg \kGuard} \cdot f~. \tag{$\dagger$}
        \end{align*}
The preexpectation ($\dagger$) is to be evaluated in some state $\sigma$ for which we have two cases:
The first case is when $\wp{C}{\iverson{\kGuard}}(\sigma) < 1$.
Using the closed form of the geometric series,
        i.e.\ $\sum_{i=0}^{\omega} q  = \frac{1}{1-q}$ if $|q| < 1$, we get
        \begin{align*}
                 & \iverson{\kGuard}(\sigma) \cdot  \wp{C}{\iverson{\neg \kGuard} \cdot f}(\sigma)
                     \cdot \left. \sum_{i=0}^{\omega} \middle( \wp{C}{\iverson{\kGuard}}(\sigma)^i \right) 
                    + \iverson{\neg \kGuard}(\sigma) \cdot f(\sigma) \tag{$\dagger$ instantiated in $\sigma$}\\
        & \eeq  \kGuardB (\sigma) \cdot 
                     \frac{\wp{C}{\iverson{\neg \kGuard} \cdot f}(\sigma)}{1-\wp{C}{\iverson{\kGuard}}(\sigma)}
                     + \iverson{\neg \kGuard}(\sigma) \cdot f(\sigma)~. \tag{closed form of geometric series}
        \end{align*}
The second case is when $\wp{C}{\iverson{\kGuard}}(\sigma) = 1$.
This case is technically slightly more involved.
The full proof can be found in Appendix~\ref{proof:thm:newrule-wpsemantics}
\qed
\end{proof} 

\noindent              
%
%
%
%
%
%
%
%
%
%
%
%
%
%
%
%
We now derive a similar proof rule for the ERT of an $f$--i.i.d.\ loop $\WHILEDO{\kGuard}{C}$.
%
%
%
\begin{theorem}[Proof Rule for ERTs of $\boldsymbol{f}$--i.i.d.\ Loops]
       \label{thm:newrule-proof-rule}
       Let $C \in \pgcl$, $\kGuard$ be a guard, and $f \in \E$ such that all of the following conditions hold:
       \begin{enumerate}
              \item $\WHILEDO{\kGuard}{C}$ is $f$--i.i.d.
              \item $\wp{C}{1} = 1$ (loop body terminates almost--surely).
              \item $\RelNewRule{\ert{C}{0}}{C}$ (every iteration runs in the same expected time).
       \end{enumerate}
       Then for the ERT of the loop $\WHILEDO{\kGuard}{C}$ w.r.t.\  postruntime $f$ it holds
       that
       \begin{align*}
              &\ert{\WHILEDO{\kGuard}{C}}{f} \eeq 1+ \frac{\kGuardB \cdot \left(1 + \ert{C}{\nkGuardB \cdot f} \right)}{1- \wp{C}{\kGuardB}} + \nkGuardB \cdot f~,
       \end{align*}
       where we define $\frac{0}{0} \coloneqq 0$ and $\frac{a}{0} \coloneqq \infty$, for $a \neq 0$.
\end{theorem}
\begin{proof}
	We first prove 
	\begin{align*}
		\ert{\WHILEDO{\kGuard}{C}}{0}  \eeq 1 + \iverson{\kGuard} \cdot \frac{1 + \ert{C}{0}}{1 - \wp{C}{\iverson{\kGuard}}}~. \tag{$\ddag$}
	\end{align*}
	To this end, we propose the following expression as the orbit of the $\ertsymbol$--char{\-}ac{\-}ter{\-}is{\-}tic function of
        the loop
        w.r.t.\ $0$:
        \begin{align*}
              F_0^n(0) \eeq 1 + \kGuardB \cdot \left( \ert{C}{0} \cdot
              \sum\limits_{i=0}^{n} \wp{C}{\kGuardB}^i ~+~ 
              \sum\limits_{i=0}^{n-1} \wp{C}{\kGuardB}^i \right)
       \end{align*}
        %
        %
        For a verification that the above expression is indeed the correct orbit, we refer to the rigorous proof of this theorem located in Appendix~\ref{lem:newrule-omega-invariant}.
        Now, analogously to the reasoning in the proof of Theorem~\ref{thm:newrule-wpsemantics} (i.e.\ using the closed form of the geometric series and case distinction on whether $\wp{C}{\kGuardB} < 1$ or $\wp{C}{\kGuardB} = 1$), we get that the supremum of this orbit is indeed the right--hand side of ($\ddag$).
        To complete the proof, consider the following:
        \belowdisplayskip=0pt
        \begin{align*}
              &\ert{\WHILEDO{\kGuard}{C}}{f} \\
              & \eeq \ert{\WHILEDO{\kGuard}{C}}{0} + \wp{\WHILEDO{\kGuard}{C}}{f} \tag{by Theorem~\ref{thm:ert-wp}} \\
              & \eeq 1 + \iverson{\kGuard} \cdot \frac{1 + \ert{C}{0}}{1 - \wp{C}{\iverson{\kGuard}}} + \iverson{\kGuard} \cdot 
                              \frac{\wp{C}{\iverson{\neg \kGuard} \cdot f}}{1-\wp{C}{\iverson{\kGuard}}}
                              + \iverson{\neg \kGuard} \cdot f \tag{by  ($\ddag$) and Theorem~\ref{thm:newrule-wpsemantics}}\\
              & \eeq 1 + \iverson{\kGuard} \cdot \frac{1 + \ert{C}{\iverson{\neg \kGuard} \cdot f}}{1 - \wp{C}{\iverson{\kGuard}}}
                              + \iverson{\neg \kGuard} \cdot f \tag{by Theorem~\ref{thm:ert-wp}}
       \end{align*}
       \normalsize
\qed
\end{proof}
%

\section{A Programming Language for Bayesian Networks}
\label{sec:applications}

\newcommand{\varXI}{x_{i}}
\newcommand{\lBody}{B_{\varXI}}
\newcommand{\psiB}{\iverson{\psi}} 
\newcommand{\npsiB}{\iverson{\neg \psi}}
\newcommand{\Blk}{\textit{Seq}}
\newcommand{\Obs}{O}

So far we have derived proof rules for formal reasoning about expected outcomes and expected run-times of i.i.d.\ loops
(Theorems~
\ref{thm:newrule-wpsemantics} and~\ref{thm:newrule-proof-rule}).
In this section, we apply these results to develop a syntactic $\pgcl$ fragment that allows exact computations of closed forms of ERTs.
In particular, no invariants, (super)martingales or fixed point computations are required.

After that, we show how BNs with observations can be translated into $\pgcl$ programs within this fragment.
Consequently, we call our $\pgcl$ fragment the \emph{Bayesian Network Language}.
As a result of the above translation, we obtain a systematic and automatable approach to compute the \emph{expected sampling time} of a BN in the presence of observations.
That is, the expected time it takes to obtain a single sample that satisfies all observations.

\subsection{The Bayesian Network Language}\label{sec:bnl}

Programs in the Bayesian Network Language are organized as sequences of blocks.
Every block is associated with a single variable, say $x$, and satisfies two constraints:
First, no variable other than $x$ is modified inside the block, i.e.\ occurs on the left--hand side of a random assignment.
Second, every variable accessed inside of a guard has been initialized before.
These restrictions ensure that there is no data flow across multiple executions of the same block.
Thus, intuitively, all loops whose body is composed from blocks (as described above) \mbox{are $f$--i.i.d.\ loops}.
\begin{definition}[The Bayesian Network Language]
    Let $\Vars = \{x_1,\, x_2,\, \ldots\}$ be a finite set of program variables as in Section~\ref{sec:pprogs}.
	The set of programs in Bayesian Network Language, denoted $\bnl$, is given by the grammar
	\begin{align*}
        C \quad\longrightarrow\quad & \Blk \quad|\quad \REPEATUNTIL{\Blk}{\psi} \quad|\quad \COMPOSE{C}{C}\\[0.5em]
        \Blk \quad\longrightarrow\quad & \COMPOSE{\Blk}{\Blk} \quad|\quad B_{x_1} \quad|\quad B_{x_2} \quad|\quad \ldots \qquad \\[0.5em]
		B_{x_i} \quad\longrightarrow\quad & \PASSIGN{x_i}{\mu} \quad|\quad \ITE{\varphi}{\PASSIGN{x_i}{\mu}}{B_{x_i}} \qquad \tag{rule exists for all $x_i \in \Vars$}
	\end{align*}
	where $x_i \in \Vars$ is a program variable, all variables in $\varphi$ have been initialized before, and 
    $B_{x_i}$ is a non--terminal parameterized with program variable $x_i \in \Vars$.
    That is, for all $x_i \in \Vars$ there is a non--terminal $B_{x_i}$.
    Moreover, $\psi$ is an arbitrary guard and $\mu$ is a distribution expression of the form 
    $\mu = \sum_{j=1}^{n} p_j \cdot \langle a_j \rangle$ with $a_j \in \Rats$ for $1 \leq j \leq n$.
\end{definition}
\begin{example}\label{ex:dice}
  Consider the \bnl program $C_{\textit{dice}}$:
  \begin{align*}
          & \COMPOSE{\PASSIGN{x_1}{\mathtt{Unif}[1\ldots 6]}}{\REPEATUNTIL{\PASSIGN{x_2}{\mathtt{Unif}[1\ldots 6]}}{x_2 \geq x_1}}
  \end{align*}
  This program first throws a fair die. 
  After that it keeps throwing a second die until its result is at least as large as the first die.
  \hfill $\triangle$
\end{example}
\noindent
For any $C \in \bnl$, our goal is to 
compute the exact ERT of $C$, i.e.\ $\ert{C}{0}$.
In case of loop--free programs, this amounts to a straightforward application of the $\ertsymbol$ calculus presented in Section~\ref{sec:pprogs}.
To deal with loops, however, we have to perform fixed point computations or require user--supplied artifacts, e.g.~invariants, supermartingales, etc.
For $\bnl$ programs, on the other hand, it suffices to apply the proof rules developed in Section~\ref{sec:ert-rules}.
As a result, we directly obtain an exact closed form solution for the ERT of a loop.
This is a consequence of 
the fact that 
all loops in $\bnl$ are $f$--i.i.d., which we establish in the following.

By definition, every loop in $\bnl$ is of the form $\REPEATUNTIL{\lBody}{\psi}$, which is equivalent to $\COMPOSE{\lBody}{\WHILEDO{\neg \psi}{\lBody}}$.
Hence, we want to apply Theorem~\ref{thm:newrule-proof-rule} to that while loop.
Our first step is to discharge the theorem's premises:
\begin{lemma}
\label{lem:apply-rule-relation}
  Let $\Blk$ be a sequence of \bnl--blocks, $g \in \E$, and $\psi$ be a guard. Then:
  \begin{enumerate}
          \item 
                The expected value of $g$ after executing $\Blk$ is unaffected by $\Blk$.
                That is, $\RelNewRule{\wp{\Blk}{g}}{\Blk}$.
          \item 
                The ERT of $\Blk$ is unaffected by $\Blk$, i.e.
                $\RelNewRule{\ert{\Blk}{0}}{\Blk}$.
          \item For every $f \in \E$, the loop $\WHILEDO{\neg \psi}{\Blk}$ is $f$--i.i.d.
  \end{enumerate}
\end{lemma}

\begin{proof}
1.~is proven by induction on the length of the sequence of blocks $\Blk$ and 2.~is a consequence of 1.,
see Appendix~\ref{app:apply-rule-relation}.
3.~follows immediately from 1.~by instantiating $g$ with $\npsiB$ and $\psiB \cdot f$, respectively.
\qed
\end{proof}
\noindent
We are now in a position to derive a closed form for the ERT of loops in $\bnl$.
\begin{theorem}\label{thm:bnl:repeat}
       For every loop $\REPEATUNTIL{\Blk}{\psi} \in \bnl$ and every $f \in \E$,
       \begin{align*}
            \ert{\REPEATUNTIL{\Blk}{\psi}}{f}
            \eeq
            \frac{1 + \ert{\Blk}{\psiB \cdot f}}{\wp{\Blk}{\psiB}}~.
       \end{align*}
       %
       %
\end{theorem}
\begin{proof}
Let $f \in \E$.
Moreover, recall that $\REPEATUNTIL{\Blk}{\psi}$ is equivalent to the program $\COMPOSE{\Blk}{\WHILEDO{\neg\psi}{\Blk}} \in \bnl$.
Applying the semantics of $\ertsymbol$ (Table~\ref{table:ert}), we proceed as follows:
\begin{flalign*}
        &\ert{\REPEATUNTIL{\Blk}{\psi}}{f}
        \eeq \ert{\Blk}{\ert{\WHILEDO{\neg \psi}{\Blk}}{f}} &
\end{flalign*}
Since the loop body $\Blk$ is loop--free, it terminates certainly, i.e.\ $\wp{\Blk}{1} = 1$ (Premise 2.~of Theorem~\ref{thm:newrule-proof-rule}).
Together with Lemma~\ref{lem:apply-rule-relation}.1.~and~3., all premises of Theorem~\ref{thm:newrule-proof-rule} are satisfied.
Hence, we obtain a closed form for $\ert{\WHILEDO{\neg \psi}{\Blk}}{f}$:
\begin{flalign*}
        &{}={}~ \ertsymbol\llbracket{\Blk}\rrbracket \biggl({~~\underbrace{1 + \frac{\npsiB \cdot \left(1 + \ert{\Blk}{\psiB \cdot f} \right)}{1-\wp{\Blk}{\npsiB}} + \psiB \cdot f}_{\eqqcolon g}~~}\biggr) &
\end{flalign*}
By Theorem~\ref{thm:ert-wp}, we know $\ert{\Blk}{g} = \ert{\Blk}{0} + \wp{C}{g}$ for any $g$. Thus:
\belowdisplayskip=0pt
\begin{flalign*}
        &{}={}~ \ert{\Blk}{0} + \wpsymbol\llbracket{\Blk}\rrbracket\biggl({~~\underbrace{1 + \frac{\npsiB \cdot \left(1 + \ert{\Blk}{\psiB \cdot f} \right)}{1-\wp{\Blk}{\npsiB}} + \psiB \cdot f}_g~~}\biggr) &
\end{flalign*}
\normalsize
Since $\wpsymbol$ is linear (Theorem~\ref{thm:basic-prop} (2)), we obtain:
\begin{flalign*}
        &{}={}~ \ert{\Blk}{0} + \underbrace{\wp{\Blk}{1}}_{\eeq 1} {} + \wp{\Blk}{\psiB \cdot f} & \\
        & \qquad + \wp{\Blk}{\frac{\npsiB \cdot \left(1 + \ert{\Blk}{\psiB \cdot f} \right)}{1-\wp{\Blk}{\npsiB}}}
\end{flalign*}
By a few simple algebraic transformations, this coincides with:
\begin{flalign*}
        &{}={}~ 1 + \ert{\Blk}{0} + \wp{\Blk}{\psiB \cdot f} + \wp{\Blk}{\npsiB \cdot \frac{1 + \ert{\Blk}{\psiB \cdot f}}{1 - \wp{\Blk}{\npsiB}}} &
\end{flalign*}
Let $R$ denote the fraction above.
Then Lemma~\ref{lem:apply-rule-relation}.1.~and~2.~implies $\RelNewRule{R}{\Blk}$.
We may thus apply Lemma~\ref{lem:newrule-main} to \mbox{derive $\wp{\Blk}{\npsiB \cdot R} = \wp{\Blk}{\npsiB} \cdot R$. Hence}:
\begin{flalign*}
        &{}={}~ 1 + \ert{\Blk}{0} + \wp{\Blk}{\psiB \cdot f} + \wp{\Blk}{\npsiB} \cdot \frac{1 + \ert{\Blk}{\psiB \cdot f}}{1 - \wp{\Blk}{\npsiB}} &
\end{flalign*}
Again, by Theorem~\ref{thm:ert-wp}, we know that $\ert{\Blk}{g} = \ert{\Blk}{0} + \wp{\Blk}{g}$ for any $g$.
Thus, for $g = \psiB \cdot f$, this yields:
\begin{flalign*}
        &{}={}~ 1 + \ert{\Blk}{\psiB \cdot f} + \wp{\Blk}{\npsiB} \cdot \frac{1 + \ert{\Blk}{\psiB \cdot f}}{1 - \wp{\Blk}{\npsiB}} &
\end{flalign*}
Then a few algebraic transformations lead us to the claimed ERT:
\begin{flalign*}
        & \eeq \frac{1 + \ert{\Blk}{\psiB \cdot f}}{\wp{\Blk}{\psiB}}~. & \tag*{\qed}
\end{flalign*}
\end{proof}
\noindent
Note that Theorem~\ref{thm:bnl:repeat} holds for arbitrary postexpectations $f \in \E$.
This enables \emph{compositional reasoning} about ERTs of $\bnl$ programs.
Since all other rules of the $\ertsymbol$--calculus for loop--free programs amount to simple syntactical transformations (see Table~\ref{table:ert}), we conclude that
\begin{corollary}\label{thm:bnl:compositional}
For any $C \in \bnl$, a closed form for $\ert{C}{0}$ can be computed compositionally.
\end{corollary}
\begin{example}
  Theorem~\ref{thm:bnl:repeat} allows us to comfortably compute the ERT of the \bnl program 
  $C_{\textit{dice}}$ introduced in Example~\ref{ex:dice}:
  \begin{align*}
          & \COMPOSE{\PASSIGN{x_1}{\mathtt{Unif}[1\ldots 6]}}{
                \REPEATUNTIL{\PASSIGN{x_2}{\mathtt{Unif}[1\ldots 6]}}{x_2 \geq x_1}
                }
  \end{align*}
  For the ERT, we have
  \begin{align*}
          & \ert{C_\textit{dice}}{0} \\
          %
          \eeq & \ert{\PASSIGN{x_1}{\mathtt{Unif}[1\ldots 6]}}{\ert{\REPEATUNTIL{\ldots}{\iverson{x_2 \geq x_1}}}{0}} 
                 \tag{Table~\ref{table:ert}} \\
          \eeq & \ert{\PASSIGN{x_1}{\mathtt{Unif}[1{\ldots} 6]}}{
                    \frac{
                        1 + \ert{\PASSIGN{x_2}{\mathtt{Unif}[1{\ldots} 6]}}{\iverson{x_2 \geq x_1}}
                    }{
                        \wp{\PASSIGN{x_1}{\mathtt{Unif}[1{\ldots} 6]}}{\iverson{x_2 \geq x_1}}
                    }
                  }
                 \tag{Thm.~\ref{thm:bnl:repeat}} \\
          \eeq & \sum_{1 \leq i \leq 6} \nicefrac 1 6 \cdot \frac{1 + \sum_{1 \leq j \leq 6} \nicefrac 1 6 \cdot \iverson{j \geq i}}{\sum_{1 \leq j \leq 6} \nicefrac 1 6 \cdot \iverson{j \geq i}} 
                 \tag{Table~\ref{table:ert}} \\
          \eeq & 3.45~. \tag*{$\triangle$} 
  \end{align*}
\end{example}

\subsection{Bayesian Networks}
To reason about expected sampling times of Bayesian networks, it remains
to develop a sound translation from BNs with observations into equivalent $\bnl$ programs.
Before we present a formal translation, let us briefly recall BNs.
A BN is a probabilistic graphical model that is given by a directed acyclic graph.
Every node is a random variable and a directed edge between two nodes expresses a probabilistic dependency between these nodes.

As a running example, consider the BN depicted in Figure~\ref{fig:bn:example} (inspired by~\cite{DBLP:books/daglib/0023091}) that
models the mood of students after taking an exam.
The network contains four random variables. They represent the difficulty of the exam ($D$), the level of preparation of a student ($P$), the achieved grade ($G$), and the resulting mood ($M$).
For simplicity, let us assume that each random variable assumes either $0$ or $1$.
The underlying dependencies express that the mood of a student depends on the achieved grade which, in turn, depends on the difficulty of the exam and the amount of preparation before taking it.
Every node is accompanied by a conditional probability table that provides the probabilities of a node given the values of all the nodes it depends upon.
We can then use the BN to answer queries such as "What is the probability that a student is well--prepared for an exam ($P = 1$), but ends up with a bad mood ($M=0$)?"

\begin{figure}[t]
\begin{tikzpicture}

\node (d) at (-2,0) [ellipse, draw, fill=gray!20, minimum width=2.2cm, minimum height=1.0cm
] {\scriptsize Difficulty};
\node (i) at (2,0) [ellipse, draw, fill=gray!20, minimum width=2.2cm, minimum height=1.0cm
] {\scriptsize Preparation};
\node (g) at (0,-2) [ellipse, draw, fill=gray!20, minimum width=2.2cm, minimum height=1.0cm
] {\scriptsize Grade};
\node (m) at (0,-4) [ellipse, draw, fill=gray!20, minimum width=2.2cm, minimum height=1.0cm
] {\scriptsize Mood};

\draw[line width=1.0pt] (d) edge[->] (g) (i)edge[->](g) (g) edge[->](m);

\node at (-5,0) {\begin{tabular}{|c | c|} \hline
\rowcolor{gray!20}\scriptsize $D=0$ & \scriptsize $D=1$  \\ \hline
\scriptsize $0.6$ &\scriptsize  $0.4$ \\
\hline
 \end{tabular}};

 \node at (5,0) {\begin{tabular}{|c | c|} \hline
\rowcolor{gray!20}\scriptsize $P=0$ & \scriptsize $P=1$  \\ \hline
\scriptsize $0.7$ &\scriptsize  $0.3$ \\
\hline
 \end{tabular}};

  \node at (-4.15,-2.2) {\begin{tabular}{ | c | c | c | } \hline
\rowcolor{gray!20} \phantom{\scriptsize ja}&\scriptsize $G=0$&\scriptsize $G=1$ \\ \hline
\scriptsize $D=0$, $P=0$ & \scriptsize $0.95$ & \scriptsize $0.05$ \\
\hline
\scriptsize $D=1$, $P=1$ & \scriptsize $0.05$ & \scriptsize $0.95$ \\
\hline
\scriptsize $D=0$, $P=1$ & \scriptsize $0.5$ & \scriptsize $0.5$ \\
\hline
\scriptsize $D=1$, $P=0$ & \scriptsize $0.6$ & \scriptsize $0.4$ \\
\hline
\end{tabular}};

\node at (3.5,-4) {\begin{tabular}{ | c | c | c | } \hline
\rowcolor{gray!20} \phantom{\scriptsize ja}&\scriptsize $M=0$&\scriptsize $M=1$ \\ \hline
\scriptsize $G=0$ & \scriptsize $0.9$ & \scriptsize $0.1$ \\
\hline
\scriptsize $G=1$ & \scriptsize $0.3$ & \scriptsize $0.7$ \\
\hline
\end{tabular}};
 
\end{tikzpicture}
\caption{A Bayesian network}
\label{fig:bn:example}
\end{figure}
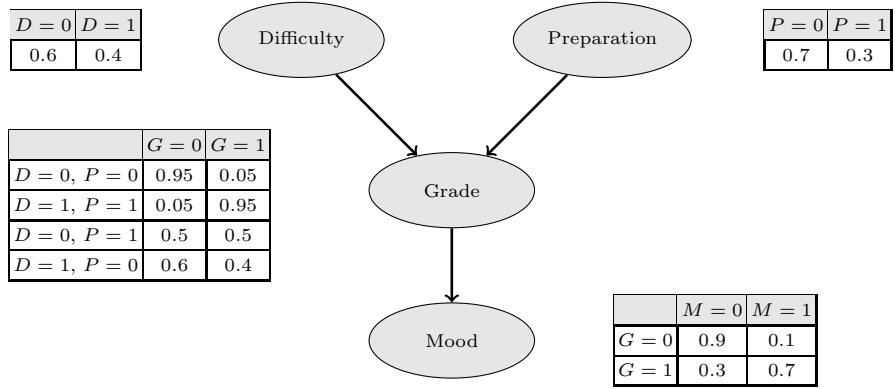

In order to translate BNs into equivalent \bnl programs, we need a formal representation first.
Technically, we consider \emph{extended} BNs in which nodes may additionally depend on inputs that are not represented by nodes in the network.
This allows us to define a compositional translation without modifying conditional probability tables.

Towards a formal definition of extended BNs, we use the following notation.
A tuple $(s_1,\ldots,s_k) \in S^{k}$ of length $k$ over some set $S$ is denoted by $\T{s}$. 
The empty tuple is $\EMPTYT$.
Moreover, for $1 \leq i \leq k$, the $i$-th element of tuple $\T{s}$ is given by $\El{\T{s}}{i}$.
To simplify the presentation, we assume that all nodes and all inputs are represented by natural numbers.

\begin{definition}\label{def:bayesian-network}
    An \emph{extended Bayesian network}, $\EBNS$ for short,  is a tuple $\EBN = (\NODES,\INPUTS,\EDGES,\VALUES,\DEP,\CPTSYM)$, where
    \begin{itemize}
        \item $\NODES \subseteq \Nats$ and $\INPUTS \subseteq \Nats$ are finite disjoint sets of \emph{nodes} and \emph{inputs}.
        \item $\EDGES \subseteq \NODES \times \NODES$ is a set of \emph{edges} such that $(\NODES,\EDGES)$ is a directed acyclic graph.
        \item $\VALUES$ is a finite set of possible \emph{values} that can be assigned to each node.
        \item $\DEP \colon \NODES \to (\NODES \cup \INPUTS)^{*}$ is a function assigning each node $v$ to an ordered sequence of \emph{dependencies}. That is,
                $\DEP(v) \eeq (u_{1}, \ldots, u_{m})$ such that $u_i < u_{i+1}$ ($1 \leq i < m$).
                Moreover, every dependency $u_j$ $(1 \leq j \leq m$) is either an input, i.e.\ $u_j \in \INPUTS$, or 
                a node with an edge to $v$, i.e.\ $u_j \in \NODES$ and $(u_j,v) \in \EDGES$.
      \item $\CPTSYM$ is a function mapping each node $v$ to its \emph{conditional probability table} $\CPT{v}$.
            That is, for $k = |\DEP(v)|$, $\CPT{v}$ is given by a function of the form
            \begin{align*}
                    \CPT{v} \,\colon\, \VALUES^{k} \To \VALUES \To [0,1]
                    \quad\text{such that}\quad
                    \sum_{\T{z} \in \VALUES^{k}, a \in \VALUES} \CPT{v}(\T{z})(a) \eeq 1. 
            \end{align*}
            Here, the $i$-th entry in a tuple $\T{z} \in \VALUES^{k}$ corresponds to the value assigned to the $i$-th entry in the sequence of dependencies $\DEP(v)$.
    \end{itemize}
    A \emph{Bayesian network} (BN) is an extended BN without inputs, i.e.\ $\INPUTS = \emptyset$.
    In particular, the dependency function is of the form $\DEP \colon \NODES \to \NODES^{*}$.
\end{definition}
\begin{example}
The formalization of our example BN (Figure~\ref{fig:bn:example}) is straightforward.
For instance, the dependencies of variable $G$ are given by $\DEP(G) = (D,P)$ (assuming $D$ is encoded by an integer less than $P$).
Furthermore, every entry in the conditional probability table of node $G$ corresponds to an evaluation of the function $\CPT{G}$.
For example, if $D = 1$, $P = 0$, and $G = 1$, we have $\CPT{G}(1,0)(1) = 0.4$.\hfill $\triangle$
\end{example}
\noindent
In general, the conditional probability table $\CPTSYM$ determines the conditional probability distribution of each node $v \in \NODES$ given the nodes and inputs it depends on.
Formally, we interpret an entry in a conditional probability table as follows:
\begin{align*}
    \PROB{v=a \,|\, \DEP(v) = \T{z}} \eeq \CPT{v}(\T{z})(a)~,
\end{align*}
where $v \in \NODES$ is a node, $a \in \VALUES$ is a value, and $\T{z}$ is a tuple of values of length $|\DEP(v)|$.
Then, by the chain rule, the joint probability of a BN is given by the product of its conditional probability tables (cf.~\cite{bishop}).

\begin{definition}\label{def:joint-distribution}
  Let $\BN = (\NODES,\INPUTS,\EDGES,\VALUES,\DEP,\CPTSYM)$ be an extended Bayesian network.
  Moreover, let $W \subseteq \NODES$ be a downward closed\footnote{$W$ is downward closed if $v \in W$ and $(u,v) \in \EDGES$ implies $u \in \EDGES$.} set of nodes.
  With each $w \in W \cup \INPUTS$, we associate a fixed value $\BVAL{w} \in \VALUES$. This notation is lifted pointwise to tuples of nodes and inputs.
  Then the \emph{joint probability} in which nodes in $W$ assume values $\BVAL{W}$ is given by
  \begin{align*}
    \PROB{ W = \BVAL{W} }
    \eeq \prod_{v \in W} \PROB{v = \BVAL{v} \,|\, \DEP(v) = \BVAL{\DEP(v)}}
    \eeq \prod_{v \in W} \CPT{v}(\BVAL{\DEP(v)})(\BVAL{v})~.
  \end{align*}
  The conditional joint probability distribution of a set of nodes $W$, given observations on a set of nodes $O$, is then given by the quotient $\nicefrac{\PROB{W = \BVAL{W}}}{\PROB{O = \BVAL{O}}}$.
\end{definition}
\noindent
For example, the probability of a student having a bad mood, i.e.\ $M = 0$, after getting a bad grade ($G=0$) for an easy exam ($D=0$) given that she was well--prepared, i.e.\ $P = 1$, is
\begin{align*}
        \PROB{D{=}0,G{=}0,M{=}0 ~|~ P {=} 1} & \eeq \frac{\PROB{D{=}0,G{=}0,M{=}0,P{=}1}}{\PROB{P {=} 1}} \\
                                             & \eeq \frac{0.9 \cdot 0.5 \cdot 0.6 \cdot 0.3}{0.3} \eeq 0.27~.
\end{align*}

\subsection{From Bayesian Networks to \textbf{\textsf{BNL}}}

We now develop a compositional translation from EBNs into \bnl programs.
Throughout this section, let $\BN = (\NODES,\INPUTS,\EDGES,\VALUES,\DEP,\CPTSYM)$ be a fixed EBN.
Moreover, with every node or input $v \in \NODES \cup \INPUTS$ we associate a program variable $\TVAR{v}$.

We proceed in three steps:
First, \emph{every node together with its dependencies} is translated into a \emph{block} of a $\bnl$ program.
These blocks are then composed into a single $\bnl$ program that captures the whole BN.
Finally, we implement conditioning by means of rejection sampling.

\paragraph{Step 1:}
We first present the atomic building blocks of our translation.
Let $v \in \NODES$ be a node.
Moreover, let $\T{z} \in \VALUES^{|\DEP(v)|}$ be an evaluation of the dependencies of $v$.
That is, $\T{z}$ is a tuple that associates a value with every node and input that $v$ depends on (in the same order as $\DEP(v)$).
For every node $v$ and evaluation of its dependencies $\T{z}$, we define a corresponding guard and a random assignment:
\begin{align*}
         \TGUARD{\BN}{v}{\T{z}} & \eeq \bigwedge_{1 \leq i \leq |\DEP(v)|} \TVAR{\El{\DEP(v)}{i}} = \El{\T{z}}{i} \\
        \TASSIGN{\BN}{v}{\T{z}} & \eeq \PASSIGN{\TVAR{v}}{ \sum_{a \in \VALUES} \CPT{v}(\T{z})(a) \cdot \langle a \rangle }
\end{align*}
Note that $\El{\DEP(v)}{i}$ is the $i$-th element from the sequence of nodes $\DEP(v)$.
\begin{example}
  Continuing our previous example (see Figure~\ref{fig:intro:bn}), assume we fixed the node $v = G$.
  Moreover, let $\T{z} = (1,0)$ be an evaluation of $\DEP(v) = (S,R)$.
  Then the guard and assignment corresponding to $v$ and $\T{z}$ are given by:
  \begin{align*}
           \TGUARD{\BN}{G}{(1,0)} \eeq & \TVAR{D} = 1 ~\wedge \TVAR{P} = 0,~\text{and} \\
          \TASSIGN{\BN}{G}{(1,0)} \eeq & \PASSIGN{\TVAR{G}}{ 0.6 \cdot \langle 0 \rangle + 0.4 \cdot \langle 1 \rangle}.
          \tag*{$\triangle$}
  \end{align*}
\end{example}
\noindent
We then translate every node $v \in \NODES$ into a program block that uses guards to determine the rows in the conditional probability table under consideration.
After that, the program samples from the resulting probability distribution using the previously constructed assignments.
In case a node does neither depend on other nodes nor input variables we omit the guards.
Formally, 
\belowdisplayskip=0pt
\begin{align*}
\TBLOCK{\BN}{v} \eeq 
  \begin{cases}
    ~~~~~~~~ \TASSIGN{\BN}{v}{\varepsilon} & ~\text{if}~ |\DEP(v)| = 0 \\[1em]
    \begin{array}{l}
    	\IF{\TGUARD{\BN}{v}{\T{z_1}}} \\
	\qquad \TASSIGN{\BN}{v}{\T{z_1}}\} \\
	\ELSE \{ \IF{\TGUARD{\BN}{v}{\T{z_2}}} \\
	\qquad \TASSIGN{\BN}{v}{\T{z_2}} \}\\
	~ \ldots \} \ELSE \{ \\
	\qquad \TASSIGN{\BN}{v}{\T{z_m}} \} \ldots \}
    \end{array} & \begin{array}{l}\text{if}~ |\DEP(v)| = k > 0 \\ ~\text{and}~ \VALUES^{k} = \{\T{z_1},\ldots,\T{z_m}\}~.\end{array}
  \end{cases}
\end{align*}
\normalsize
\begin{remark}
  The guards under consideration are conjunctions of equalities between variables and literals.
  We could thus use a more efficient translation of conditional probability tables by adding
  a \texttt{switch-case} statement to our probabilistic programming language. Such a statement is of the form
  \begin{align*}
          \texttt{switch}(\T{x})~\{~\texttt{case}~\T{a}_1: C_1~\texttt{case}~a_2: C_2~\ldots~\texttt{default}: C_m \}~,
  \end{align*}
  where $\T{x}$ is a tuple of variables, and $\T{a}_1, \ldots \T{a}_{m-1}$ are tuples of rational numbers of the same length as $\T{x}$.
  With respect to the $\wpsymbol$ semantics, a \texttt{switch-case} statement is syntactic sugar for nested \texttt{if-then-else} blocks as used in the above translation.
  However, the runtime model of a \texttt{switch-case} statement requires just a single guard evaluation ($\varphi$) instead of potentially multiple guard evaluations when evaluating nested \texttt{if-then-else} blocks.
  Since the above adaption is straightforward, we opted to use nested \texttt{if-then-else} blocks to keep our programming language simple and allow, in principle, more general guards.
  \hfill $\triangle$
\end{remark}
\paragraph{Step 2}
The next step is to translate a complete EBN into a $\bnl$ program. 
To this end, we compose the blocks obtained from each node starting at the roots of the network.
That is, all nodes that contain no incoming edges.
Formally,
\begin{align*}
  \ROOTS{\BN} = \{ v \in \NODES_{\BN} ~|~ \neg \exists u \in \NODES_{\BN} \colon (u,v) \in \EDGES_{\BN} \}~.
\end{align*}
After translating every node in the network, we remove them from the graph, i.e.\ every root becomes an input, and proceed with the translation
until all nodes have been removed.
More precisely, given a set of nodes $S \subseteq \NODES$, the extended BN $\BN \setminus S$ obtained by removing $S$ from $\BN$ is defined as
\begin{align*}
        \BN \setminus S \eeq
        \left( \NODES \setminus S,\, \INPUTS \cup S,\, \EDGES \setminus (\NODES \times S \cup S \times \NODES),\, \DEP,\, \CPTSYM \right)~.
\end{align*}
With these auxiliary definitions readily available, an extended BN $\BN$ is translated into a $\bnl$ program as follows:
\begin{align*}
  \TB{\BN}\eeq 
  \begin{cases}
    \TBLOCK{\BN}{r_1};\ldots;\TBLOCK{\BN}{r_m} & ~\text{if}~ \ROOTS{\BN} = \{r_1,\ldots,r_m\} = \NODES \\
    \TBLOCK{\BN}{r_1};\ldots;\TBLOCK{\BN}{r_m}; & ~\text{if}~ \ROOTS{\BN} = \{r_1,\ldots,r_m\} \subsetneqq \NODES  \\
    \TB{\BN \setminus \ROOTS{\BN}} 
  \end{cases}
\end{align*}
\paragraph{Step 3}
To complete the translation, it remains to account for observations.
Let $\BNCOND : \NODES \to \VALUES \cup \{ \bot \}$
be a function mapping every node either to an observed value in $\VALUES$ or to $\bot$.
The former case is interpreted as an observation that node $v$ has value $\BNCOND(v)$.
Otherwise, i.e.\ if $\BNCOND(v) = \bot$, the value of node $v$ is \emph{not observed}.
We collect all observed nodes in the set $\Obs = \{ v \in \NODES ~|~ \BNCOND(v) \neq \bot \}$. 
It is then natural to incorporate conditioning into our translation by applying rejection sampling: We repeatedly execute a \bnl program until every observed node has the desired value $\BNCOND(v)$.
In the presence of observations, we translate the extended BN $\BN$ into a $\bnl$ program as follows:
\belowdisplayskip=0pt
\begin{align*}
    \TBN{\BN}{\BNCOND} \eeq \REPEATUNTIL{\TB{\BN}}{\bigwedge_{v \in \Obs} \TVAR{v} = \BNCOND(v)}
\end{align*}
\normalsize
%
\begin{figure}[t]
    \centering
    \begin{minipage}{0.3\textwidth}
        \begin{align*}
                1\quad&\REPEAT \\
                2\quad&\qquad \PASSIGN{\TVAR{D}}{0.6 \cdot \valueIn{0} + 0.4 \cdot \valueIn{1}}; \\
                3\quad&\qquad \PASSIGN{\TVAR{P}}{0.7 \cdot \valueIn{0} + 0.3 \cdot \valueIn{1}} \\
                4\quad&\qquad \IFBRACK{\TVAR{D}=0 \wedge \TVAR{P}=0} \\
                5\quad&\qquad \qquad \PASSIGN{\TVAR{G}}{0.95 \cdot \valueIn{0} + 0.05 \cdot \valueIn{1}} \\
                6\quad&\qquad \} \ELSE \IFBRACK{\TVAR{D}=1 \wedge \TVAR{P}=1} \\
                7\quad&\qquad \qquad \PASSIGN{\TVAR{G}}{0.05 \cdot \valueIn{0} + 0.95 \cdot \valueIn{1}} \\
                8\quad&\qquad \} \ELSE \IFBRACK{\TVAR{D}=0 \wedge \TVAR{P}=1} \\
                9\quad&\qquad \qquad \PASSIGN{\TVAR{G}}{0.5 \cdot \valueIn{0} + 0.5 \cdot \valueIn{1}}
        \end{align*}
    \end{minipage}
    ~
    \begin{minipage}{0.5\textwidth}
        \begin{align*}
               10\quad&\qquad \} \ELSE \{ \\
               11\quad&\qquad \qquad \PASSIGN{\TVAR{G}}{0.6 \cdot \valueIn{0} + 0.4 \cdot \valueIn{1}} \\
               12\quad&\qquad \}; \\
               13\quad&\qquad \IFBRACK{\TVAR{G}=0} \\
               14\quad&\qquad \qquad \PASSIGN{\TVAR{M}}{0.9 \cdot \valueIn{0} + 0.1 \cdot \valueIn{1}} \\
               15\quad&\qquad \} \ELSE \{ \\
               16\quad&\qquad \qquad \PASSIGN{\TVAR{M}}{0.3 \cdot \valueIn{0} + 0.7 \cdot \valueIn{1}} \\
               17\quad&\qquad \} \\
               18\quad&\UNTIL{\TVAR{P}=1}
        \end{align*}
    \end{minipage}
    \caption{The $\bnl$ program $C_{\textit{mood}}$ obtained from the BN in Figure~\ref{fig:bn:example}.}
\label{fig:bn:translation}
\end{figure}
\begin{example}
  Consider, again, the BN $\BN$ depicted in Figure~\ref{fig:bn:example}.
  Moreover, assume we observe $P = 1$. 
  Hence, the conditioning function $\BNCOND$ is given by $\BNCOND(P) = 1$ and $\BNCOND(v) = \bot$ for $v \in \{D,G,M\}$.
  Then the translation of $\BN$ and $\BNCOND$, i.e.\ $\TBN{\BN}{\BNCOND}$, is the $\bnl$ program $C_{\textit{mood}}$ depicted in Figure~\ref{fig:bn:translation}.\hfill $\triangle$
\end{example}
\noindent
Since our translation yields a $\bnl$ program for any given BN, we can compositionally compute a closed form for the expected simulation time of a BN.
This is an immediate consequence of Corollary~\ref{thm:bnl:compositional}.

We still have to prove, however, that our translation is sound, i.e.\ the conditional joint probabilities inferred from a BN coincide 
with the (conditional) joint probabilities from the corresponding $\bnl$ program.
Formally, we obtain the following soundness result.
\begin{theorem}[Soundness of Translation]
\label{thm:soundness-bayesian-network-to-bnl}
Let $\BN = (\NODES,\INPUTS,\EDGES,\VALUES,\DEP,\CPTSYM)$ be a BN and $\BNCOND : \NODES \to \VALUES \,\cup\, \{\bot\}$ be a function determining the observed nodes.
For each node and input $v$, let $\BVAL{v} \in \VALUES$ be a fixed value associated with $v$. 
In particular, we set $\BVAL{v} = \BNCOND(v)$ for each observed node $v \in \Obs$.
Then
\begin{align*}
        \wp{\TBN{\BN}{\BNCOND}}{\iverson{\bigwedge_{v \in \NODES \setminus O} \TVAR{v} = \BVAL{v}}}
        \eeq
        \frac{
            \PROB{ \bigwedge_{v \in \NODES} v = \BVAL{v} }
        }{
            \PROB{ \bigwedge_{o \in O} o = \BVAL{o} }
        }
        ~. 
\end{align*}
\end{theorem}
\begin{proof}
Without conditioning, i.e.\ $\Obs = \emptyset$, the proof proceeds by induction on the number of nodes of $\BN$.
With conditioning, we additionally apply Theorems~\ref{thm:newrule-wpsemantics} and~\ref{thm:bnl:repeat} to deal with loops introduced by observed nodes.
See Appendix~\ref{proof-thm:soundness-bayesian-network-to-bnl}.
\qed
\end{proof}
%

%
%
\begin{example}[Expected Sampling Time of a BN]
Consider, again, the BN $\BN$ in Figure~\ref{fig:bn:example}.
Moreover, recall the corresponding program $C_{\textit{mood}}$ derived from $\BN$ in Figure~\ref{fig:bn:translation}, where we observed $P=1$.
By Theorem~\ref{thm:soundness-bayesian-network-to-bnl} we can also determine the probability that a student who got a bad grade in an easy exam was well--prepared by means of weakest precondition reasoning.
This yields
\begin{align*}
        & \wp{C_{\textit{mood}}}{\iverson{\TVAR{D}=0 \wedge \TVAR{G} = 0 \wedge \TVAR{M} = 0}} \\
        & \eeq
        \frac{\PROB{D=0,G=0,M=0,P=1}}{\PROB{P = 1}} 
        \eeq 0.27~.
\end{align*}
Furthermore, by Corollary~\ref{thm:bnl:compositional}, it is straightforward to determine the expected time to obtain a single sample of $\BN$ that satisfies the observation $P = 1$:
\begin{align*}
    \ert{C_{\textit{mood}}}{0} 
    \eeq \frac{
            1 + \ert{C_{\textit{loop-body}}}{0}
         }{
            \wp{C_{\textit{loop-body}}}{\iverson{P=1}}
         }
    \eeq 23.4 + \nicefrac{1}{15}
    \eeq 23.4\bar{6}~.
    \tag*{$\triangle$}
\end{align*}
\end{example}


\section{Implementation}
\label{sec:implementation}

We implemented a prototype in \textsf{Java} to analyze expected sampling times of Bayesian networks.
More concretely, our tool takes as input a BN together with observations in the popular Bayesian Network Interchange Format.\footnote{\url{http://www.cs.cmu.edu/~fgcozman/Research/InterchangeFormat/}}
The BN is then translated into a \bnl program as shown in Section~\ref{sec:applications}.
Our tool applies the $\ertsymbol$--calculus together with our proof rules developed in Section~\ref{sec:ert-rules} to compute the exact expected runtime of the \bnl program.

The size of the resulting \bnl program is linear in the total number of rows of all conditional probability tables in the BN.
The program size is thus \emph{not} the bottleneck of our analysis.
As we are dealing with an NP--hard problem~\cite{DBLP:journals/ai/Cooper90,DBLP:journals/ai/DagumL93}, it is not surprising that our algorithm has a worst--case exponential time complexity.
However, also the space complexity of our algorithm is exponential in the worst case:
As an expectation is propagated backwards through an $\texttt{if}$--clause of the \bnl program, the size of the expectation is potentially multiplied.
This is also the reason that our analysis runs out of memory on some benchmarks.

We evaluated our implementation on the \emph{largest} BNs in the Bayesian Network Repository~\cite{scutari2012bayesian}
that consists --- to a large extent --- of real--world BNs including expert systems for, e.g., 
electromyography (\texttt{munin})~\cite{andreassen1989munin},
hematopathology diagnosis (\texttt{hepar2})~\cite{onisko1998probabilistic},
weather forecasting (\texttt{hailfinder})~\cite{abramson1996hailfinder}, and printer troubleshooting in Windows 95 (\texttt{win95pts})~\cite[Section 5.6.2]{ramanna2013emerging}.
For a evaluation of \emph{all} BNs in the repository, see Appendix.

All experiments were performed on an HP BL685C G7. 
Although up to 48 cores with 2.0GHz were available, only one core was used apart from \textsf{Java}'s garbage collection.
The \textsf{Java} virtual machine was limited to 8GB of RAM.

Our experimental results are shown in Table~\ref{tbl:experiments}.
\begin{table}[t]
\center
\begin{adjustbox}{max width=\textwidth}
\renewcommand{\arraystretch}{1.1}
\begin{tabular}{lr@{\quad}r@{\quad}l@{\quad}|@{\quad}r@{\quad}r@{\quad}l@{\quad}|@{\quad}r@{\quad}r@{\quad}l}
        \textbf{BN} 
        & \textbf{\#obs} & \textbf{Time} & \textbf{EST} 
        & \textbf{\#obs} & \textbf{Time} & \textbf{EST} 
        & \textbf{\#obs} & \textbf{Time} & \textbf{EST} \\
    \hline
    \hline
    \texttt{earthquake} & \multicolumn{9}{l}{\emph{\#nodes: 5, \#edges: 4, avg. Markov Blanket: 2.00}} \\
    \hline
    & $0$ & $0.09$ & $8.000 \cdot 10^0$ 
    & $2$ & $0.23$ & $9.276 \cdot 10^1$
    & $4$ & $0.24$ & $1.857 \cdot 10^2$ \\
    \hline
    \texttt{cancer} & \multicolumn{9}{l}{\emph{\#nodes: 5, \#edges: 4, avg. Markov Blanket: 2.00}} \\
    \hline
    & $0$ & $0.09$ & $8.000 \cdot 10^0$ 
    & $2$ & $0.22$ & $1.839 \cdot 10^1$
    & $5$ & $0.20$ & $5.639 \cdot 10^2$ \\
    \hline
    \texttt{survey} & \multicolumn{9}{l}{\emph{\#nodes: 6, \#edges: 6, avg. Markov Blanket: 2.67}} \\
    \hline
    & $0$ & $0.09$ & $1.000 \cdot 10^1$ 
    & $2$ & $0.21$ & $2.846 \cdot 10^2$
    & $5$ & $0.22$ & $9.113 \cdot 10^3$ \\
    \hline
    \texttt{asia} & \multicolumn{9}{l}{\emph{\#nodes: 8, \#edges: 8, avg. Markov Blanket: 2.50}} \\
    \hline
    & $0$ & $0.26$ & $1.400\cdot 10^1$ 
    & $2$ & $0.25$ & $3.368 \cdot 10^2$
    & $6$ & $0.25$ & $8.419 \cdot 10^4$ \\
    \hline
    \texttt{sachs} & \multicolumn{9}{l}{\emph{\#nodes: 11, \#edges: 17, avg. Markov Blanket: 3.09}} \\
    \hline
    & $0$ & $0.13$ & $2.000\cdot 10^1$ 
    & $3$ & $0.24$ & $7.428 \cdot 10^2$
    & $6$ & $2.72$ & $5.533 \cdot 10^7$ \\
    \hline
    \texttt{insurance} & \multicolumn{9}{l}{\emph{\#nodes: 27, \#edges: 52, avg. Markov Blanket: 5.19}} \\
    \hline
    & $0$ & $0.17$ & $5.200\cdot 10^1$ 
    & $3$ & $0.31$ & $5.096 \cdot 10^3$
    & $5$ & $0.91$ & $1.373 \cdot 10^5$ \\
    \hline
    \texttt{alarm} & \multicolumn{9}{l}{\emph{\#nodes: 37, \#edges: 46, avg. Markov Blanket: 3.51}} \\
    \hline
    & $0$ & $0.14$ & $6.200 \cdot 10^1$ 
    & $2$ & MO & ---
    & $6$ & $40.47$ & $3.799 \cdot 10^5$ \\
    \hline
    \texttt{barley} & \multicolumn{9}{l}{\emph{\#nodes: 48, \#edges: 84, avg. Markov Blanket: 5.25}} \\
    \hline
    & $0$ & $0.46$ & $8.600 \cdot 10^1$ 
    & $2$ & $0.53$ & $5.246 \cdot 10^4$ 
    & $5$ & MO & --- \\
    \hline
    \texttt{hailfinder} & \multicolumn{9}{l}{\emph{\#nodes: 56, \#edges: 66, avg. Markov Blanket: 3.54}} \\
    \hline
    & $0$ & $0.23$ & $9.500 \cdot 10^1$ 
    & $5$ & $0.63$ & $5.016 \cdot 10^5$ 
    & $9$ & $0.46$ & $9.048 \cdot 10^6$ \\
    \hline
    \texttt{hepar2} & \multicolumn{9}{l}{\emph{\#nodes: 70, \#edges: 123, avg. Markov Blanket: 4.51}} \\
    \hline
    & $0$ & $0.22$ & $1.310 \cdot 10^2$ 
    & $1$ & $1.84$ & $1.579 \cdot 10^2$ 
    & $2$ & MO & --- \\
    \hline
    \texttt{win95pts} & \multicolumn{9}{l}{\emph{\#nodes: 76, \#edges: 112, avg. Markov Blanket: 5.92}} \\
    \hline
    & $0$ & $0.20$ & $1.180 \cdot 10^2$ 
    & $1$ & $0.36$ & $2.284 \cdot 10^3$ 
    & $3$ & $0.36$ & $4.296 \cdot 10^5$ \\
    & $7$ & $0.91$ & $1.876 \cdot 10^6$ 
    & $12$ & $0.42$ & $3.973 \cdot 10^7$ 
    & $17$ & $61.73$ & $1.110 \cdot 10^{15}$ \\
    \hline
    \texttt{pathfinder} & \multicolumn{9}{l}{\emph{\#nodes: 135, \#edges: 200, avg. Markov Blanket: 3.04}} \\
    \hline
    & $0$ & $0.37$ & $217$ 
    & $1$ & $0.53$ & $1.050 \cdot 10^4$ 
    & $3$ & $31.31$ & $2.872 \cdot 10^4$ \\
    & $5$ & MO & ---
    & $7$ & $5.44$ & $\infty$ 
    & $7$ & $480.83$ & $\infty$ \\
    \hline
    \texttt{andes} & \multicolumn{9}{l}{\emph{\#nodes: 223, \#edges: 338, avg. Markov Blanket: 5.61}} \\
    \hline
    & $0$ & $0.46$ & $3.570 \cdot 10^2$ 
    & $1$ & MO & --- 
    & $3$ & $1.66$ & $5.251 \cdot 10^3$ \\
    & $5$ & $1.41$ & $9.862 \cdot 10^3$ 
    & $7$ & $0.99$ & $8.904 \cdot 10^4$ 
    & $9$ & $0.90$ & $6.637 \cdot 10^5$ \\
    \hline
    \texttt{pigs} & \multicolumn{9}{l}{\emph{\#nodes: 441, \#edges: 592, avg. Markov Blanket: 3.66}} \\
    \hline
    & $0$ & $0.57$ & $7.370 \cdot 10^2$ 
    & $1$ & $0.74$ & $2.952 \cdot 10^3$ 
    & $3$ & $0.88$ & $2.362 \cdot 10^3$ \\
    & $5$ & $0.85$ & $1.260 \cdot 10^5$ 
    & $7$ & $1.02$ & $1.511 \cdot 10^6$ 
    & $8$ & MO & --- \\
    \hline
    \texttt{munin} & \multicolumn{9}{l}{\emph{\#nodes: 1041, \#edges: 1397, avg. Markov Blanket: 3.54}} \\
    \hline
    & $0$ & $1.29$ & $1.823 \cdot 10^3$ 
    & $1$ & $1.47$ & $3.648 \cdot 10^4$ 
    & $3$ & $1.37$ & $1.824 \cdot 10^7$ \\
    & $5$ & $1.43$ & $\infty$ 
    & $9$ & $1.79$ & $1.824 \cdot 10^{16}$ 
    & $10$ & $65.64$ & $1.153 \cdot 10^{18}$ \\
    \hline
\end{tabular}
\renewcommand{\arraystretch}{1}
\end{adjustbox}
\vspace{.25\baselineskip}
\caption{Experimental results. Time is in seconds. MO\ denotes out of memory.}
\label{tbl:experiments}
\end{table}
The number of nodes of the considered BNs ranges from 56 to 1041.
For each Bayesian network, we computed the expected sampling time (EST) for different collections of observed nodes (\#obs).
Furthermore, Table~\ref{tbl:experiments} provides the \emph{average Markov Blanket size}, i.e. the average number of parents, children and children's parents of nodes in the BN~\cite{pearl1985bayesian},
as an indicator measuring how independent nodes in the BN are.

Observations were picked at random.
Note that the time required by our prototype varies depending on both the number of observed nodes and the actual observations.
Thus, there are cases in which we run out of memory although the total number of observations is small.

In order to obtain an understanding of what the EST corresponds to in actual execution times on a real machine, 
we also performed simulations for the \texttt{win95pts} network.
More precisely, we generated \textsf{Java} programs from this network analogously to the translation in Section~\ref{sec:applications}.
This allowed us to approximate that our \textsf{Java} setup can execute $9.714\cdot 10^6$ steps (in terms of EST) per second.

For the \texttt{win95pts} with 17 observations, an EST of $1.11 \cdot 10^{15}$ then corresponds to an expected time of approximately \emph{$3.6$ years} in order to obtain a \emph{single} valid sample.
We were additionally able to find a case with 13 observed nodes where our tool discovered within 0.32 seconds an EST that corresponds to approximately $4.3$ \emph{million years}.
In contrast, exact inference using variable elimination was almost instantaneous.
This demonstrates that knowing expected sampling times upfront can indeed be beneficial when selecting an inference method.

\section{Conclusion}\label{sec:conclusion}

We presented a syntactic notion of independent and identically distributed probabilistic loops and derived dedicated proof rules to determine exact expected outcomes and runtimes of such loops.
These rules do not require any user--supplied information, such as invariants, (super)martingales, etc.

Moreover, we isolated a syntactic fragment of probabilistic programs that allows to compute expected runtimes in a highly automatable fashion.
This fragment is non--trivial:
We show that all Bayesian networks can be translated into programs within this fragment.
Hence, we obtain an automated formal method for computing expected simulation times of Bayesian networks.
We implemented this method and successfully applied it to various real--world BNs that stem from, amongst others, medical applications.
Remarkably, our tool was capable of proving extremely large expected sampling times within seconds.

There are several directions for future work: 
For example, there exist subclasses of BNs for which exact inference is in $\textsf{P}$, e.g. polytrees. 
Are there analogies for probabilistic programs?
Moreover, it would be interesting to consider more complex graphical models, such as recursive BNs~\cite{DBLP:journals/jacm/EtessamiY09}.


\bibliographystyle{splncs03}
\bibliography{bibliography}

\begin{thebibliography}{10}
\providecommand{\url}[1]{\texttt{#1}}
\providecommand{\urlprefix}{URL }

\bibitem{abramson1996hailfinder}
Abramson, B., Brown, J., Edwards, W., Murphy, A., Winkler, R.L.: {Hailfinder: A
  Bayesian System for Forecasting Severe Weather}. International Journal of
  Forecasting  12(1),  57--71 (1996)

\bibitem{andreassen1989munin}
Andreassen, S., Jensen, F.V., Andersen, S.K., Falck, B., Kj{\ae}rulff, U.,
  Woldbye, M., S{\o}rensen, A., Rosenfalck, A., Jensen, F.: {MUNIN: An Expert
  EMG Assistant}. In: Computer--aided Electromyography and Expert Systems, pp.
  255--277. Pergamon Press (1989)

\bibitem{bishop}
Bishop, C.: Pattern Recognition and Machine Learning. Springer (2006)

\bibitem{DBLP:conf/rta/BournezG05}
Bournez, O., Garnier, F.: {Proving Positive Almost--Sure Termination}. In:
  {RTA}. LNCS, vol. 3467, pp. 323--337. Springer (2005)

\bibitem{DBLP:journals/jcss/BrazdilKKV15}
Br{\'{a}}zdil, T., Kiefer, S., Kucera, A., Varekov{\'{a}}, I.H.: {Runtime
  Analysis of Probabilistic Programs with Unbounded Recursion}. J. Comput.
  Syst. Sci.  81(1),  288--310 (2015)

\bibitem{McIver:FM:2005}
Celiku, O., McIver, A.: {Compositional Specification and Analysis of
  Cost--based Properties in Probabilistic Programs}. In: FM. LNCS, vol. 3582,
  pp. 107--122. Springer (2005)

\bibitem{DBLP:conf/cav/ChakarovS13}
Chakarov, A., Sankaranarayanan, S.: {Probabilistic Program Analysis with
  Martingales}. In: {CAV}. LNCS, vol. 8044, pp. 511--526. Springer (2013)

\bibitem{DBLP:conf/popl/ChatterjeeFNH16}
Chatterjee, K., Fu, H., Novotn{\'{y}}, P., Hasheminezhad, R.: {Algorithmic
  Analysis of Qualitative and Quantitative Termination Problems for Affine
  Probabilistic Programs}. In: {POPL}. pp. 327--342. {ACM} (2016)

\bibitem{DBLP:conf/popl/ChatterjeeNZ17}
Chatterjee, K., Novotn{\'{y}}, P., Zikelic, D.: {Stochastic Invariants for
  Probabilistic Termination}. In: {POPL}. pp. 145--160. {ACM} (2017)

\bibitem{DBLP:journals/kbs/ConstantinouFN12}
Constantinou, A.C., Fenton, N.E., Neil, M.: {pi--football: A Bayesian Network
  Model for Forecasting Association Football Match Outcomes}. Knowl.--Based
  Syst.  36,  322--339 (2012)

\bibitem{DBLP:journals/ai/Cooper90}
Cooper, G.F.: {The Computational Complexity of Probabilistic Inference Using
  Bayesian Belief Networks}. Artif. Intell.  42(2-3),  393--405 (1990)

\bibitem{DBLP:journals/ai/DagumL93}
Dagum, P., Luby, M.: {Approximating Probabilistic Inference in Bayesian Belief
  Networks is NP--Hard}. Artif. Intell.  60(1),  141--153 (1993)

\bibitem{DBLP:journals/cacm/Dijkstra75}
Dijkstra, E.W.: {Guarded Commands, Nondeterminacy and Formal Derivation of
  Programs}. Communications of the {ACM}  18(8),  453--457 (1975)

\bibitem{DBLP:books/ph/Dijkstra76}
Dijkstra, E.W.: A Discipline of Programming. Prentice--Hall (1976)

\bibitem{DBLP:journals/jacm/EtessamiY09}
Etessami, K., Yannakakis, M.: {Recursive Markov Chains, Stochastic Grammars,
  and Monotone Systems of Nonlinear Equations}. JACM  56(1),  1:1--1:66 (2009)

\bibitem{luis}
Fioriti, L.M.F., Hermanns, H.: {Probabilistic Termination: Soundness,
  Completeness, and Compositionality}. In: POPL. pp. 489--501. ACM (2015)

\bibitem{DBLP:conf/recomb/FriedmanLNP00}
Friedman, N., Linial, M., Nachman, I., Pe'er, D.: {Using Bayesian Networks to
  Analyze Expression Data}. In: {RECOMB}. pp. 127--135. {ACM} (2000)

\bibitem{DBLP:conf/cade/FrohnNHBG16}
Frohn, F., Naaf, M., Hensel, J., Brockschmidt, M., Giesl, J.: {Lower Runtime
  Bounds for Integer Programs}. In: {IJCAR} 2016. pp. 550--567 (2016)

\bibitem{DBLP:conf/popl/Goodman13}
Goodman, N.D.: {The Principles and Practice of Probabilistic Programming}. In:
  {POPL}. pp. 399--402. {ACM} (2013)

\bibitem{DBLP:conf/uai/GoodmanMRBT08}
Goodman, N.D., Mansinghka, V.K., Roy, D.M., Bonawitz, K., Tenenbaum, J.B.:
  {Church: A Language for Generative Models}. In: {UAI}. pp. 220--229. {AUAI}
  Press (2008)

\bibitem{DBLP:conf/popl/GordonGRRBG14}
Gordon, A.D., Graepel, T., Rolland, N., Russo, C.V., Borgstr{\"{o}}m, J.,
  Guiver, J.: {Tabular: A Schema--Driven Probabilistic Programming Language}.
  In: {POPL}. pp. 321--334. {ACM} (2014)

\bibitem{DBLP:conf/icse/GordonHNR14}
Gordon, A.D., Henzinger, T.A., Nori, A.V., Rajamani, S.K.: {Probabilistic
  Programming}. In: Future of Software Engineering. pp. 167--181. {ACM} (2014)

\bibitem{DBLP:series/sci/Heckerman08}
Heckerman, D.: {A Tutorial on Learning with Bayesian Networks}. In: Innovations
  in Bayesian Networks, Studies in Computational Intelligence, vol. 156, pp.
  33--82. Springer (2008)

\bibitem{Hehner:FAC:2011}
Hehner, E.C.R.: {A Probability Perspective}. Formal Aspects of Computing
  23(4),  391--419 (2011)

\bibitem{hoffman2014no}
Hoffman, M.D., Gelman, A.: {The No--U--turn Sampler: Adaptively Setting Path
  Lengths in Hamiltonian Monte Carlo}. Journal of Machine Learning Research
  15(1),  1593--1623 (2014)

\bibitem{DBLP:journals/jamia/JiangC10}
Jiang, X., Cooper, G.F.: {A Bayesian Spatio--temporal Method for Disease
  Outbreak Detection}. {JAMIA}  17(4),  462--471 (2010)

\bibitem{DBLP:conf/mfcs/KaminskiK15}
Kaminski, B.L., Katoen, J.: {On the Hardness of Almost--Sure Termination}. In:
  MFCS. LNCS, vol. 9234, pp. 307--318. Springer (2015)

\bibitem{lics17}
Kaminski, B.L., Katoen, J.: {A Weakest Pre--expectation Semantics for
  Mixed--sign Expectations}. In: {LICS} (2017)

\bibitem{DBLP:conf/esop/KaminskiKMO16}
Kaminski, B.L., Katoen, J., Matheja, C., Olmedo, F.: {Weakest Precondition
  Reasoning for Expected Run--times of Probabilistic Programs}. In: {ESOP}.
  LNCS, vol. 9632, pp. 364--389. Springer (2016)

\bibitem{DBLP:books/daglib/0023091}
Koller, D., Friedman, N.: Probabilistic Graphical Models - Principles and
  Techniques. {MIT} Press (2009)

\bibitem{DBLP:journals/jcss/Kozen81}
Kozen, D.: {S}emantics of {P}robabilistic {P}rograms. J. Comput. Syst. Sci.
  22(3),  328--350 (1981)

\bibitem{DBLP:journals/jcss/Kozen85}
Kozen, D.: {A Probabilistic PDL}. J. Comput. Syst. Sci.  30(2),  162--178
  (1985)

\bibitem{DBLP:journals/ipl/LassezNS82}
Lassez, J.L., Nguyen, V.L., Sonenberg, L.: {Fixed Point Theorems and Semantics:
  A Folk Tale}. Information Processing Letters  14(3),  112--116 (1982)

\bibitem{mciver}
McIver, A., Morgan, C.: {A}bstraction, {R}efinement and {P}roof for
  {P}robabilistic {S}ystems. Springer (2004)

\bibitem{infernet}
Minka, T., Winn, J.: {Infer.NET} (2017),
  \url{http://infernet.azurewebsites.net/}, {Accessed online October 17}

\bibitem{DBLP:conf/nips/MinkaW08}
Minka, T., Winn, J.M.: Gates. In: {NIPS}. pp. 1073--1080. Curran Associates
  (2008)

\bibitem{DBLP:conf/sas/Monniaux01}
Monniaux, D.: {An Abstract Analysis of the Probabilistic Termination of
  Programs}. In: SAS. LNCS, vol. 2126, pp. 111--126. Springer (2001)

\bibitem{neapolitan2010probabilistic}
Neapolitan, R.E., Jiang, X.: Probabilistic Methods for Financial and Marketing
  Informatics. Morgan Kaufmann (2010)

\bibitem{DBLP:conf/aaai/NoriHRS14}
Nori, A.V., Hur, C., Rajamani, S.K., Samuel, S.: {R2: An Efficient MCMC Sampler
  for Probabilistic Programs}. In: {AAAI}. pp. 2476--2482. {AAAI} Press (2014)

\bibitem{DBLP:conf/lics/OlmedoKKM16}
Olmedo, F., Kaminski, B.L., Katoen, J., Matheja, C.: {Reasoning about Recursive
  Probabilistic Programs}. In: {LICS}. pp. 672--681. {ACM} (2016)

\bibitem{onisko1998probabilistic}
Onisko, A., Druzdzel, M.J., Wasyluk, H.: {A Probabilistic Causal Model for
  Diagnosis of Liver Disorders}. In: Proceedings of the Seventh International
  Symposium on Intelligent Information Systems (IIS--98). pp. 379--387 (1998)

\bibitem{pearl1985bayesian}
Pearl, J.: {Bayesian Networks: A Model of Self--activated Memory for Evidential
  Reasoning}. In: Proc.\ of CogSci. pp. 329--334 (1985)

\bibitem{pfeffer2009figaro}
Pfeffer, A.: {Figaro: An Object--oriented Probabilistic Programming Language}.
  Charles River Analytics Technical Report  137, ~96 (2009)

\bibitem{ramanna2013emerging}
Ramanna, S., Jain, L.C., Howlett, R.J.: Emerging paradigms in machine learning.
  Springer (2013)

\bibitem{scutari2012bayesian}
Scutari, M.: {Bayesian Network Repository}. URL http://www.bnlearn.com  (2017)

\bibitem{DBLP:journals/siamcomp/SharirPH84}
Sharir, M., Pnueli, A., Hart, S.: {Verification of Probabilistic Programs}.
  {SIAM} J. Comput.  13(2),  292--314 (1984)

\bibitem{DBLP:conf/aistats/WoodMM14}
Wood, F., van~de Meent, J., Mansinghka, V.: {A New Approach to Probabilistic
  Programming Inference}. In: {AISTATS}. {JMLR} Workshop and Conference
  Proceedings, vol.~33, pp. 1024--1032. JMLR.org (2014)

\bibitem{DBLP:journals/mcm/YuanD06}
Yuan, C., Druzdzel, M.J.: {Importance Sampling Algorithms for Bayesian
  Networks: Principles and Performance}. Mathematical and Computer Modelling
  43(9-10),  1189--1207 (2006)

\bibitem{DBLP:conf/aaai/ZweigR98}
Zweig, G., Russell, S.J.: {Speech Recognition with Dynamic Bayesian Networks}.
  In: {AAAI/IAAI}. pp. 173--180. {AAAI} Press / The {MIT} Press (1998)

\end{thebibliography}

\clearpage
\appendix
\section{Appendix}

\subsection{Calculation for Expected Sampling Time of Example Network}\label{app:example:calculations}
\newcommand{\vNodeR}{x_{\nodeR}}
\newcommand{\vNodeS}{x_{\nodeS}}
\newcommand{\vNodeG}{x_{\nodeG}}
\newcommand{\bnlExampleLoopBody}{C}
Let $\BN$ denote the Bayesian Network shown in Figure~\ref{fig:intro:bn}. We denote the nodes by $\nodeR$, $\nodeS$, and $\nodeG$, respectively. 
Since we want to condition $\nodeG$ to value $0$, let $\BNCOND (\nodeG)=0$, $\BNCOND (\nodeR)=\bot$, and $\BNCOND (\nodeS)=\bot$.
According to the translation from Bayesian Networks to $\bnl$ programs presented in Section~\ref{sec:bnl}, we obtain the following
program $\TBN{\BN}{\BNCOND}$:
\begin{align*}
       &\mathtt{repeat} \{ \\
       &\qquad //~ \text{translation of node $\nodeR$, denoted $\TNODE{\BN}{\nodeR}{\BNCOND}$} \\
       &\qquad \PASSIGN{\vNodeR}{a \cdot \langle 0 \rangle + (1-a) \cdot \langle 1 \rangle};\\
       &\qquad \phantom{aa} \\
       &\qquad //~ \text{translation of node $\nodeS$, denoted $\TNODE{\BN}{\nodeS}{\BNCOND}$} \\
       &\qquad \IFBRACK{\vNodeR = 0} \\
       &\qquad \qquad \PASSIGN{\vNodeS}{a \cdot \langle 0 \rangle + (1-a) \cdot \langle 1 \rangle} \\ 
       &\qquad \} ~\ELSE~ \{ \\
       &\qquad \qquad \PASSIGN{\vNodeS}{0.2 \cdot \langle 0 \rangle + 0.8 \cdot \langle 1 \rangle} \}; \\
       &\qquad \phantom{aa} \\
       &\qquad //~ \text{translation of node $\nodeG$, denoted $\TNODE{\BN}{\nodeG}{\BNCOND}$} \\
       &\qquad \IFBRACK{\vNodeS =0 \wedge \vNodeR =0} \\ 
       &\qquad \qquad  \PASSIGN{\vNodeG}{0.01 \cdot \langle 0 \rangle + (0.99 \cdot \langle 1 \rangle} \\  
       &\qquad \}~ \ELSE \IFBRACK{\vNodeS =0 \wedge \vNodeR =1} \\
       &\qquad \qquad  \PASSIGN{\vNodeG}{0.25 \cdot \langle 0 \rangle + (0.75 \cdot \langle 1 \rangle}\\
       &\qquad  \}~ \ELSE \IFBRACK{\vNodeS =1 \wedge \vNodeR =0} \\
       &\qquad \qquad  \PASSIGN{\vNodeG}{0.9 \cdot \langle 0 \rangle + (0.1 \cdot \langle 1 \rangle} \\
       &\qquad \}~ \ELSE~ \{ \\
       &\qquad \qquad  \PASSIGN{\vNodeG}{0.2 \cdot \langle 0 \rangle + (0.8 \cdot \langle 1 \rangle} \}\\
       &\}\mathtt{until}\left(\vNodeG = 0 \right) // \text{condition node $\nodeG$ to value $0$}
\end{align*}
Thus we have
\begin{align*}
       &\ert{\TBN{\BN}{\BNCOND}}{0} \\
       &\eeq \ert{\REPEATUNTIL{\COMPOSE{\COMPOSE{\TNODE{\BN}{\nodeR}{\BNCOND}}{\TNODE{\BN}{\nodeS}{\BNCOND}}}{\TNODE{\BN}{\nodeG}{\BNCOND}}}{\vNodeG = 0}}{0}~. 
\end{align*} 
Let 
$\bnlExampleLoopBody$ denote the body of the above loop. Theorem~\ref{thm:bnl:repeat} then yields
\begin{align*}
       \ert{\REPEATUNTIL{\bnlExampleLoopBody}{\vNodeG = 0}}{0} ~{}={}~ \frac{1+ \ert{\bnlExampleLoopBody}{0}}{\wp{\bnlExampleLoopBody}{\iverson{\vNodeG = 0}}}~.
\end{align*}
We first compute $\ert{C}{0} = \ert{\COMPOSE{\COMPOSE{\TNODE{\BN}{\nodeR}{\BNCOND}}{\TNODE{\BN}{\nodeS}{\BNCOND}}}{\TNODE{\BN}{\nodeG}{\BNCOND}}}{0}$. We have
\begin{align*}
      \ert{\TNODE{\BN}{\nodeG}{\BNCOND}}{0} ~{}={}~ 2 + \iverson{\vNodeS \neq 0 \vee \vNodeR \neq 0} \cdot \left(1 + \iverson{\vNodeS \neq 0 \vee \vNodeR \neq0} \right)~,
\end{align*}
thus
\begin{align*}
       \ert{\TNODE{\BN}{\nodeS}{\BNCOND}}{\ert{\TNODE{\BN}{\nodeG}{\BNCOND}}{0}} ~{}={}~ \frac{18}{5} + \iverson{\vNodeR = 0} \cdot \left( \frac{2}{5} - 2a \right)~,
\end{align*}
and finally
\begin{align*}
       &\ert{C}{0} \\
       &\eeq \ert{\TNODE{\BN}{\nodeR}{\BNCOND}}{\ert{\TNODE{\BN}{\nodeS}{\BNCOND}}{\ert{\TNODE{\BN}{\nodeG}{\BNCOND}}{0}}} \\
       &\eeq -2a^2 + \frac{2}{5} a + \frac{18}{5}~.
\end{align*}
Next, we determine $\wp{\bnlExampleLoopBody}{\iverson{\vNodeG = 0}}$ and therefore calculate
\begin{align*}
      \wp{\TNODE{\BN}{\nodeG}{\BNCOND}}{\iverson{\vNodeG = 0}} 
      &\eeq 0.01 \cdot \iverson{\vNodeS =0 \wedge \vNodeR =0}  + 0.25 \cdot \iverson{\vNodeS = 0 \wedge \vNodeR = 1} \\
      &\qquad + 0.9 \cdot \iverson{\vNodeS = 1 \wedge \vNodeR = 0}
              + 0.2 \cdot \iverson{\vNodeS = 1 \wedge \vNodeR = 1}~,
\end{align*}
and
\begin{align*}
       &\wp{\TNODE{\BN}{\nodeS}{\BNCOND}}{\wp{\TNODE{\BN}{\nodeG}{\BNCOND}}{0}} \\
       &\eeq \frac{21}{100} + \iverson{\vNodeR = 0} \cdot \left(- \frac{89}{100} a + \frac{69}{100} \right)~,
\end{align*}
thus we finally obtain
\begin{align*}
       &\wp{C}{\iverson{\vNodeG = 0}} \\
       &\eeq \wp{\TNODE{\BN}{\nodeR}{\BNCOND}}{\wp{\TNODE{\BN}{\nodeS}{\BNCOND}}{\wp{\TNODE{\BN}{\nodeG}{\BNCOND}}{\iverson{\vNodeG = 0}}}} \\
       &\eeq -\frac{89}{100} a^2 + \frac{69}{100} a + \frac{21}{100}~.
\end{align*}
Overall, by applying some simple algebra, we get
\begin{align*}
       \ert{\REPEATUNTIL{\bnlExampleLoopBody}{\vNodeG = 0}}{0} ~{}={}~ \frac{200a^2 - 40a - 460}{89a^2 -69a -21}~.
\end{align*}
\subsection{Proof of Lemma~\ref{lem:newrule-main}}
\label{proof-lem:newrule-main}
\begin{proof}
       By induction on the structure of C. As the induction base we have the atomic programs: \\ \\
       \emph{$\SKIP$:} We have
       \begin{equation*}
              \wp{\SKIP}{g \cdot f} \eeq  g \cdot f \eeq  g \cdot \wp{\SKIP}{f}~. 
              \tag{by Table~\ref{table:wp}}
       \end{equation*} \\
       \emph{$\EMPTY$:} Analogous to $\SKIP$. \\ \\
       \emph{$\DIVERGE$:} We have
       \begin{equation*}
              \wp{\DIVERGE}{g \cdot f} \eeq  0 \eeq  g \cdot \wp{\DIVERGE}{f}~.
              \tag{by Table~\ref{table:wp}}
       \end{equation*}
       \emph{$\PASSIGN{x}{\mu}$:} We have
       \begin{align*}
                            \wp{\PASSIGN{x}{\mu}}{g \cdot f} \eeq & \lambda \sigma\mydot \Exp{\Rats}{\big(\lambda v\mydot (f \cdot g)\subst{x}{v}\big)}{\mu_\sigma} 
                            \tag{by Table~\ref{table:wp}}\\
             \eeq & \lambda \sigma\mydot \Exp{\Rats}{\big(\lambda v\mydot f\subst{x}{v}\big) \cdot g}{\mu_\sigma} 
                            \tag{$\RelNewRule{g}{(\PASSIGN{x}{\mu})}$ by assumption} \\
             \eeq & g \cdot \left( \lambda \sigma\mydot \Exp{\Rats}{\big(\lambda v\mydot f\subst{x}{v}\big)}{\mu_\sigma} \right) \\
             \eeq & g \cdot \wp{\PASSIGN{x}{\mu}}{f}~.
       \end{align*}
       As the induction hypothesis we now assume that for arbitrary but fixed $C_1, C_2 \in \pgcl$ and all $f,g \in E$ 
       with $\RelNewRule{g}{C_1}$ and $\RelNewRule{g}{C_2}$ it holds that both
       \begin{equation*}
              \wp{C_1}{g \cdot f} \eeq  g \cdot \wp{C_1}{f}
       \end{equation*}
       and
       \begin{equation*}
              \wp{C_2}{g \cdot f} \eeq  g \cdot \wp{C_2}{f}~.
       \end{equation*}  \\     
       \emph{$\ITE{\kGuard}{C_1}{C_2}$:} We have
       \begin{align*}
                        & \wp{\ITE{\kGuard}{C_1}{C_2}}{g \cdot f} \\
              {}={}~& \kGuardB \cdot \wp{C_1}{g \cdot f} + \nkGuardB \cdot \wp{C_2}{g \cdot f}
                           \tag{by Table~\ref{table:wp}} \\
              {}={}~& \kGuardB \cdot g \cdot \wp{C_1}{f} + \nkGuardB \cdot \wp{C_2}{g \cdot f} 
                           \tag{$\RelNewRule{g}{C_1}$ and I.H.\ on $C_1$} \\
              {}={}~&  \kGuardB \cdot g \cdot \wp{C_1}{f} + \nkGuardB \cdot g \cdot  \wp{C_2}{f} 
                           \tag{$\RelNewRule{g}{C_2}$ and I.H.\ on $C_2$} \\
              {}={}~& g \cdot \left(  \kGuardB \cdot \wp{C_1}{f} + \nkGuardB \cdot  \wp{C_2}{f} \right) \\
              {}={}~& g \cdot \wp{\ITE{\kGuard}{C_1}{C_2}}{f}~.
                           \tag{by Table~\ref{table:wp}}
       \end{align*}
       \\       
        \emph{$\COMPOSE{C_1}{C_2}$: } We have
        \begin{align*}
                               \wp{\COMPOSE{C_1}{C_2}}{g \cdot f}
                \eeq & \wp{C_1}{\wp{C_2}{g \cdot f}}
                               \tag{by Table~\ref{table:wp}} \\
                \eeq & \wp{C_1}{g \cdot \wp{C_2}{f}}
                               \tag{$\RelNewRule{g}{C_2}$ and I.H.\ on $C_2$} \\
                \eeq & g \cdot \wp{C_1}{\wp{C_2}{f}}
                               \tag{$\RelNewRule{g}{C_1}$ and I.H.\ on $C_1$} \\
                \eeq & g \cdot \wp{\COMPOSE{C_1}{C_2}}{f} ~.
                               \tag{by Table~\ref{table:wp}}
        \end{align*} \\
        \emph{$\WHILEDO{\kGuard}{C_1}$:} This case is more involved. Recall that
        \begin{equation}
               \label{eqn:newrule:lemmavars_sup}
               \wp{\WHILEDO{\kGuard}{C_1}}{h} \eeq  \sup_{n\in \Nats}~ F_h^n(0)
        \end{equation}
        holds for all $h \in \E$. Thus we show 
        \begin{equation}
               \label{eqn:newrule-lemmavars-while}
               F_{g \cdot f}^n(0) \eeq  g \cdot F_f^n(0)
        \end{equation}
        for all $n \in \Nats$, which implies the desired result. We proceed by induction on $n$. \\ \\
        \emph{Induction base $n=0$.} We have 
        \begin{align*}
                              &F_{f \cdot g}^0(0) \eeq 0 \eeq g \cdot 0 \eeq g \cdot F_{f}^0(0)~.
        \end{align*} 
        \emph{Induction Hypothesis:} Equation~\ref{eqn:newrule-lemmavars-while} holds for some arbitrary but fixed $n \in \Nats$. \\ \\
        \emph{Induction Step $n \mapsto n+1$.} We have
	\begin{align*}
        		& F_{f \cdot g}^{n+1}(0) \stackrel{!}{\eeq} g \cdot F_f^{n+1}(0) \\
        		& \text{iff}\quad  F_{f \cdot g}(F_{f \cdot g}^{n}(0)) \eeq g \cdot F_f(F_f^{n}(0)) \\
        		& \text{iff}\quad  \iverson{\neg\kGuard} \cdot f \cdot g + \wp{C_1}{F_{f \cdot g}^{n}(0)} \eeq g \cdot \left(  \iverson{\neg\kGuard} \cdot f + \wp{C_1}{F_f^{n}(0)} \right) \tag{by Definition~\ref{def:wp}}\\
        		& \text{iff}\quad \wp{C_1}{F_{f \cdot g}^{n}(0)} \eeq g \cdot \wp{C_1}{F_f^{n}(0)} \\
        		& \text{iff}\quad \wp{C_1}{g \cdot F_{f}^{n}(0)} \eeq g \cdot \wp{C_1}{F_f^{n}(0)} \tag{by I.H.}\\
        		& \text{iff}\quad g \cdot \wp{C_1}{F_{f}^{n}(0)} \eeq g \cdot \wp{C_1}{F_f^{n}(0)} \tag{by $\RelNewRule{g}{C_1}$ and I.H. on $C_1$}\\
        		& \text{iff}\quad \true~.
	\end{align*}
        \emph{$\REPEATUNTIL{\kGuard}{C_1}$:} We have
        \begin{align*}
                             & \wp{\REPEATUNTIL{C_1}{\kGuard}}{g \cdot f} \\
               &\eeq  \wp{\COMPOSE{C_1}{\WHILEDO{\kGuard}{C_1}}}{g\cdot f} \\
               &\eeq  \wp{C_1}{\wp{\WHILEDO{\kGuard}{C_1}}{g\cdot f}} 
                              \tag{by Table~\ref{table:wp}} \\
               &\eeq  \wp{C_1}{g \cdot \wp{\WHILEDO{\kGuard}{C_1}}{f}}           
                              \tag{see case $\WHILEDO{\kGuard}{C_1}$} \\
               &\eeq  g \cdot \wp{C_1}{\wp{\WHILEDO{\kGuard}{C_1}}{f}} 
                              \tag{$\RelNewRule{g}{C_1}$ and I.H.\ on $C_1$} \\
               &\eeq  g \cdot \wp{\REPEATUNTIL{C_1}{\kGuard}}{f}
                              \tag{by Table~\ref{table:wp}}~.                       
        \end{align*}
        \qed
\end{proof}
\noindent
%
%
%
%
\subsection{Proof of Lemma~\ref{thm:newrule_wp_finite_approximations}}
\label{lem:newrule_wp_finite_approximations}

\begin{proof}
       By induction on $n$. We consider the cases $n=1$ and $n=2$ as the induction base.
       For $n=1$, we have
       \begin{align*}
              F_f^1(0) 
              &\eeq  \iverson{\kGuard} \cdot \wp{C}{0} + \iverson{\neg \kGuard} \cdot f \tag{by Definition~\ref{def:wp}} \\
              &\eeq  \iverson{\kGuard} \cdot 0 + \iverson{\neg \kGuard} \cdot f \tag{by strictness, see Theorem~\ref{thm:basic-prop} (\ref{thm:basic-prop-strictness})} \\
              &\eeq  \iverson{\neg \kGuard} \cdot f \tag{$\dagger$} \\
              &\eeq  \iverson{\kGuard} \cdot 0 + \iverson{\neg \kGuard} \cdot f \\
              &\eeq  \iverson{\kGuard} \cdot \wp{C}{\iverson{\neg \kGuard} \cdot f}
               \cdot \underbrace{\left.\sum\limits_{i=0}^{-1}  \middle(\wp{C}{\iverson{\kGuard}}^i \right)}_{{}=0} + \iverson{\neg \kGuard} \cdot f~.
       \end{align*}
       For $n=2$, we have
       \begin{align*}
       F_f^2(0) &\eeq 
              \iverson{\kGuard} \cdot \wp{C}{F_f^1(0)} + \iverson{\neg \kGuard} \cdot f
              \tag{by Definition~\ref{def:wp}} \\
              &\eeq  \iverson{\kGuard} \cdot \wp{C}{\iverson{\neg \kGuard} \cdot f} + \iverson{\neg \kGuard} \cdot f
              \tag{by $\dagger$} \\
              &\eeq  \iverson{\kGuard} \cdot\wp{C}{\iverson{\neg \kGuard} \cdot f}
               \cdot \underbrace{\left.\sum\limits_{i=0}^{0}  \middle(\wp{C}{\iverson{\kGuard}}^i \right)}_{{}=1} + \iverson{\neg \kGuard} \cdot f~.
       \end{align*}
       Before we perform the inductive step, observe the following equality implied by 
       Lemma~\ref{lem:newrule-main},
       which holds for all $n \in \Nats$:
       \begin{align*}
                        &\wp{C}{\iverson{\kGuard} \cdot \wp{C}{\iverson{\neg \kGuard} \cdot f}
                           \cdot \sum\limits_{i=0}^{n-2} \wp{C}{\iverson{\kGuard}}^i} \\
              &\eeq \sum\limits_{i=0}^{n-2} \wp{C}{\iverson{\kGuard} \cdot \wp{C}{\iverson{\neg \kGuard} \cdot f}
                           \cdot \wp{C}{\iverson{\kGuard}}^i} 
                           \tag{by linearity, see Theorem~\ref{thm:basic-prop} (\ref{thm:basic-prop-linearity})} \\
              &\eeq \sum\limits_{i=0}^{n-2} \wp{C}{\iverson{\kGuard} \cdot \wp{C}{\iverson{\neg \kGuard} \cdot f}}
                           \cdot \wp{C}{\iverson{\kGuard}}^i
                           \tag{by Lemma~\ref{lem:newrule-main}, $\RelNewRule{\wp{C}{\iverson{\kGuard}}}{C}$}  \\
              &\eeq \sum\limits_{i=0}^{n-2} \wp{C}{\iverson{\kGuard}} \cdot \wp{C}{\iverson{\neg \kGuard} \cdot f}
                           \cdot \wp{C}{\iverson{\kGuard}}^i
                           \tag{by Lemma~\ref{lem:newrule-main}, $\RelNewRule{\wp{C}{\iverson{\neg \kGuard} \cdot f}}{C}$}  \\
              &\eeq \wp{C}{\iverson{\neg \kGuard} \cdot f}
       \cdot \sum\limits_{i=1}^{n-1}  \wp{C}{\iverson{\kGuard}}^i \tag{$\ddagger$}
       \end{align*}
       As induction hypothesis, we now assume that the lemma holds for some arbitrary but fixed $n \in \Nats \setminus \{ 0 \}$. 
       Then
       \begin{align*}
       &F_f^{n+1}(0) \\
       & \eeq \iverson{\kGuard} \cdot \wp{C}{F_f^n(0)} + \iverson{\neg \kGuard} \cdot f \tag{by Definition~\ref{def:wp}} \\
       &\eeq \iverson{\kGuard} \cdot \wp{C}{ \iverson{\kGuard} \cdot \wp{C}{\iverson{\neg \kGuard} \cdot f}
                    \cdot \left.\sum\limits_{i=0}^{n-2} \middle( \wp{C}{\iverson{\kGuard}}^i \right) + \iverson{\neg \kGuard} \cdot f} \\
                  & \qquad\quad + \iverson{\neg \kGuard} \cdot f
       \tag{by I.H.} \\
       &\eeq \iverson{\kGuard} {\cdot} \left( \wp{C}{\iverson{\kGuard} {\cdot} \wp{C}{\iverson{\neg \kGuard} {\cdot} f}
                    {\cdot}\sum\limits_{i=0}^{n-2}\wp{C}{\iverson{\kGuard}}^i }
                    + \wp{C}{\iverson{\neg \kGuard} {\cdot} f} \right) \\
	& \qquad\quad + \iverson{\neg \kGuard} {\cdot} f
                    \tag{by linearity, see Theorem~\ref{thm:basic-prop} (\ref{thm:basic-prop-linearity})} \\
      &\eeq  \iverson{\kGuard} {\cdot} \left(\wp{C}{\iverson{\neg \kGuard} {\cdot} f}
                     {\cdot} \left. \sum\limits_{i=1}^{n-1} \middle(  \wp{C}{\iverson{\kGuard}}^i \right)
                  + \wp{C}{\iverson{\neg \kGuard} {\cdot} f} \right) + \iverson{\neg \kGuard} {\cdot} f  
                  \tag{by Lemma~\ref{lem:newrule-main}, $\RelNewRule{\wp{C}{\iverson{\neg \kGuard} \cdot f}}{C}$, $\RelNewRule{\wp{C}{\iverson{\kGuard}}}{C}$, and $\iverson{\kGuard} \cdot \iverson{\kGuard} = \iverson{\kGuard}$} \\
       &\eeq  \iverson{\kGuard} \cdot \left(\wp{C}{\iverson{\neg \kGuard} \cdot f}
                     \cdot \sum\limits_{i=0}^{n-1}  \wp{C}{\iverson{\kGuard}}^i  \right) 
                     + \iverson{\neg \kGuard} \cdot f ~. \tag*{\qed}
       \end{align*}
\end{proof}
\noindent

\subsection{Proof of Theorem~\ref{thm:newrule-wpsemantics}}

\label{proof:thm:newrule-wpsemantics}
\begin{proof}
	We have
        \begin{align*}
                 & \wp{\WHILEDO{\kGuard}{C}}{f} \\
        & \eeq  \sup_{n \in \Nats}~ F_f^n(0) \tag{by Definition~\ref{def:wp}}
                      \\
        & \eeq  \sup_{n \in \Nats}~
                      \iverson{\kGuard} \cdot  \wp{C}{\iverson{\neg \kGuard} \cdot f}
                     \cdot \left. \sum\limits_{i=0}^{n-2} \middle( \wp{C}{\iverson{\kGuard}}^i \right) 
                    + \iverson{\neg \kGuard} \cdot f 
                    \tag{by Lemma \ref{thm:newrule_wp_finite_approximations}} \\
        & \eeq \iverson{\kGuard} \cdot  \wp{C}{\iverson{\neg \kGuard} \cdot f}
                     \cdot \left. \sum_{i=0}^{\omega} \middle( \wp{C}{\iverson{\kGuard}}^i \right) 
                    + \iverson{\neg \kGuard} \cdot f~. \tag{$\dagger$}
        \end{align*}
$\dagger$ is to be evaluated in some state $\sigma$ for which we have two cases:
The first case is when $\wp{C}{\iverson{\kGuard}}(\sigma) < 1$.
Using the closed form of the geometric series,
        i.e.\ $\sum_{i=0}^{\omega} q  = \frac{1}{1-q}$ if $|q| < 1$, we get
        \begin{align*}
                 & \iverson{\kGuard}(\sigma) \cdot  \wp{C}{\iverson{\neg \kGuard} \cdot f}(\sigma)
                     \cdot \left. \sum_{i=0}^{\omega} \middle( \wp{C}{\iverson{\kGuard}}(\sigma)^i \right) 
                    + \iverson{\neg \kGuard}(\sigma) \cdot f(\sigma) \tag{$\dagger$ instantiated in $\sigma$}\\
        & \eeq  \kGuardB (\sigma) \cdot 
                     \frac{\wp{C}{\iverson{\neg \kGuard} \cdot f}(\sigma)}{1-\wp{C}{\iverson{\kGuard}}(\sigma)}
                     + \iverson{\neg \kGuard}(\sigma) \cdot f(\sigma)~. \tag{closed form of geometric series}
        \end{align*}
The second case is when $\wp{C}{\iverson{\kGuard}}(\sigma) = 1$.
This can only be true if every state $\tau$ that is reachable with non--zero probability by executing $C$ on $\sigma$ has to satisfy $\varphi$. Otherwise, the integral would not add up to exactly 1.
Formally,
\begin{align*}
        \wp{C}{\iverson{\varphi}}(\sigma) \eeq 1 \qiff & \Exp{\Sigma}{\iverson{\varphi}}{\semantics{C}{\sigma}} = 1 \\
        \qiff & \forall \semantics{C}{\sigma}(\tau) > 0 \implies \iverson{\varphi}(\tau) = 1 \tag{$\ddagger$}
	\label{eq:bennistheorem}
\end{align*}
From that, it follows that 
\begin{align*}
	\wp{C}{\iverson{\neg \kGuard} \cdot f} (\sigma) \eeq \Exp{\Sigma}{\iverson{\neg \kGuard} \cdot f}{\semantics{C}{\sigma}} & \eeq \Exp{\Sigma}{f - \iverson{\kGuard} \cdot f}{\semantics{C}{\sigma}} \\
	& \eeq \Exp{\Sigma}{f - 1 \cdot f}{\semantics{C}{\sigma}} \tag{by $\ddagger$} \\
	& \eeq \Exp{\Sigma}{0}{\semantics{C}{\sigma}} \eeq 0 \tag{$\star$}
\end{align*}
We then get
        \begin{align*}
                 & \iverson{\kGuard}(\sigma) \cdot  \wp{C}{\iverson{\neg \kGuard} \cdot f}(\sigma)
                     \cdot \left. \sum_{i=0}^{\omega} \middle( \wp{C}{\iverson{\kGuard}}(\sigma)^i \right) 
                    + \iverson{\neg \kGuard}(\sigma) \cdot f(\sigma) \tag{instantiate $\dagger$ in $\sigma$}\\
	& \eeq \iverson{\kGuard}(\sigma) \cdot  0
                     \cdot \infty
                    + \iverson{\neg \kGuard}(\sigma) \cdot f(\sigma) \tag{by $\star$}\\
	& \eeq \iverson{\kGuard}(\sigma) \cdot  0
                    + \iverson{\neg \kGuard}(\sigma) \cdot f(\sigma) \tag{by $0 \cdot \infty = 0$}\\
	& \eeq \iverson{\kGuard}(\sigma) \cdot  \frac{0}{0}
                    + \iverson{\neg \kGuard}(\sigma) \cdot f(\sigma) \tag{by $\frac{0}{0} = 0$}\\
        & \eeq  \kGuardB (\sigma) \cdot 
                     \frac{\wp{C}{\iverson{\neg \kGuard} \cdot f}(\sigma)}{1-\wp{C}{\iverson{\kGuard}}(\sigma)}
                     + \iverson{\neg \kGuard}(\sigma) \cdot f(\sigma)~. \tag{Case 2}
        \end{align*}
For both cases we get the desired form which completes the proof.
\qed
\end{proof}
\subsection{$\boldertsymbol$ Orbits of $f$--i.i.d. Loops w.r.t.\ Postruntime $0$}\label{lem:newrule-omega-invariant}

\begin{lemma}[$\boldertsymbol$ Orbits of $\boldsymbol{0}$--i.i.d.\ Loops]
       Let $C \in \pgcl$, $\kGuard$ be a guard, $f \in \E$ such that the loop $\WHILEDO{\kGuard}{C}$ is $0$--i.i.d., and let $F_0$ be the corresponding $\ertsymbol$--characteristic function.
       Furthermore, assume 
       \begin{equation}
       \label{eqn:newrule-omega-condition-two}
              \wp{C}{1} = 1
       \end{equation} 
       and
       \begin{equation}
       \label{eqn:newrule-omega-condition-one}
              \RelNewRule{\ert{C}{0}}{C}~.
       \end{equation}
       Then the orbit of $F_0$ is given by
       \begin{align*}
              F_0^n(0) \eeq 1 + \kGuardB \cdot \left( \ert{C}{0} \cdot
              \sum\limits_{i=0}^{n-1} \wp{C}{\kGuardB}^i ~+~ 
              \sum\limits_{i=0}^{n-2} \wp{C}{\kGuardB}^i \right)
       \end{align*} 
       for $n \geq 1$.
\end{lemma}
\begin{proof}
       By induction on $n$. 
       For the induction base, we have $n=1$:
       \begin{align*}
                              F_0 (0) &\eeq  1 + \kGuardB \cdot \ert{C}{0} \\
              &\eeq  1 + \kGuardB \cdot \left( \ert{C}{0} \cdot \sum\limits_{i=0}^{0} \wp{C}{\kGuardB}^i
                              ~+~ \sum\limits_{i=0}^{-1} \wp{C}{\kGuardB}^i \right) 
       \end{align*}
       The second induction base is $n=2$:
       \begin{align*}
                             F_0^2 (0) &\eeq  1 + \kGuardB \cdot \ert{C}{1 + \kGuardB \cdot \ert{C}{0}} \\
              &\eeq  1+\kGuardB \cdot \left(\ert{C}{0} + \wp{C}{1 + \kGuardB \cdot \ert{C}{0}} \right)
              \tag{by Theorem~\ref{thm:ert-wp}} \\
              &\eeq  1+ \kGuardB \cdot \left(\ert{C}{0} + \wp{C}{1} + \wp{C}{\kGuardB \cdot \ert{C}{0}} \right)
              \tag{by linearity, see Theorem~\ref{thm:basic-prop} (\ref{thm:basic-prop-linearity})} \\
              &\eeq  1+ \kGuardB \cdot \left(\ert{C}{0} + 1 + \wp{C}{\kGuardB \cdot \ert{C}{0}} \right)
              \tag{by Equation \ref{eqn:newrule-omega-condition-two}} \\
              &\eeq  1+ \kGuardB \cdot \left(\ert{C}{0} + 1 + \wp{C}{\kGuardB} \cdot \ert{C}{0} \right)
              \tag{by Lemma~\ref{lem:newrule-main}, $\RelNewRule{\ert{C}{0}}{C}$} \\
              &\eeq  1 + \kGuardB \cdot \left( \ert{C}{0} \cdot \sum\limits_{i=0}^{1} \wp{C}{\kGuardB}^i
                              ~+~ \sum\limits_{i=0}^{0} \wp{C}{\kGuardB}^i \right) ~.              
       \end{align*}
       Before we prove the induction step, observe the following:
       Lemma~\ref{lem:newrule-main} yields
       \begin{align}
                         & \wp{C}{\kGuardB \cdot \ert{C}{0} \cdot \sum\limits_{i=0}^{n} \wp{C}{\kGuardB}^i}
                             \notag \\
             &\eeq  \sum\limits_{i=0}^{n} \wp{C}{\kGuardB \cdot \ert{C}{0} \cdot \wp{C}{\kGuardB}^i}
             \tag{by linearity, see Theorem~\ref{thm:basic-prop} (\ref{thm:basic-prop-linearity})} \\
             &\eeq \sum\limits_{i=0}^{n} \wp{C}{\kGuardB \cdot \ert{C}{0}} \cdot \wp{C}{\kGuardB}^i
             \tag{by Lemma~\ref{lem:newrule-main}, $\RelNewRule{\wp{C}{\kGuardB}}{C}$} \\
             &\eeq \sum\limits_{i=0}^{n} \wp{C}{\kGuardB} \cdot \ert{C}{0} \cdot \wp{C}{\kGuardB}^i
             \tag{by Lemma~\ref{lem:newrule-main}, $\RelNewRule{\ert{C}{0}}{C}$} \\
             &\eeq \ert{C}{0} \cdot \sum\limits_{i=1}^{n+1} \wp{C}{\kGuardB}^i~.
             \label{eqn:newrule-omega-equation-one}
       \end{align}
       Lemma \ref{lem:newrule-main} implies further
       \begin{align}
                         & \wp{C}{\kGuardB \cdot \sum\limits_{i=0}^{n} \wp{C}{\kGuardB}^i} 
                             \notag \\
             &\eeq \sum\limits_{i=0}^{n} \wp{C}{\kGuardB \cdot \wp{C}{\kGuardB}^i} 
                            \tag{by linearity, see Theorem~\ref{thm:basic-prop} (\ref{thm:basic-prop-linearity})} \\  
             &\eeq \sum\limits_{i=0}^{n} \wp{C}{\kGuardB} \cdot \wp{C}{\kGuardB}^i 
                            \tag{by Lemma \ref{lem:newrule-main}, $\RelNewRule{\wp{C}{\kGuardB}}{C}$} \\
             &\eeq \sum\limits_{i=1}^{n+1} \wp{C}{\kGuardB}^i ~.
                            \label{eqn:newrule-omega-equation-two}                               
       \end{align}
       We are now in a position two show that $F_0 (I_n) = I_{n+1}$ holds. We have
\begin{align*}
	&F_0^{n+1}(0) \\
	&\eeq  1 + \kGuardB \cdot \ert{C}{F_0^n(0)} \\
	&\eeq  1 + \kGuardB \cdot \left(\ert{C}{0} + \wp{C}{F_0^n(0)} \right) 
		\tag{by Theorem~\ref{thm:ert-wp}} \\
	&\eeq  1 + \kGuardB \cdot \Biggl(\ert{C}{0} \\
	&\qquad + \wp{C}{1 + \kGuardB \cdot \left( \ert{C}{0} 
		\cdot \sum_{i=0}^{n-1} \wp{C}{\kGuardB}^i 
		~+~ \sum_{i=0}^{n-2} \wp{C}{\kGuardB}^i \right)} \Biggr)
		\tag{by I.H. on $n$} \\
	&\eeq  1 + \kGuardB \cdot \Biggl(\ert{C}{0} \\ 
	&\qquad + \wp{C}{1 + \kGuardB \cdot  \ert{C}{0} 
		\cdot \sum_{i=0}^{n-1} \wp{C}{\kGuardB}^i 
		~+~ \kGuardB \cdot \sum_{i=0}^{n-2} \wp{C}{\kGuardB}^i} \Biggr) \\
	&\eeq  1 + \kGuardB \cdot \left( \ert{C}{0} 
	+ \wp{C}{1} 
	+ \wp{C}{\kGuardB \cdot \ert{C}{0} \cdot \sum_{i=0}^{n-1} \wp{C}{\kGuardB}^i} \right. \\
	&\left. \quad \quad \quad \quad \quad \quad 
		{}+ \wp{C}{\kGuardB \cdot \sum\limits_{i=0}^{n-2} \wp{C}{\kGuardB}^i} \right)                           
	\tag{by linearity, see Theorem~\ref{thm:basic-prop} (\ref{thm:basic-prop-linearity})} \\
	&\eeq  
	1 
	+ \kGuardB \cdot \left( \ert{C}{0} 
	+ 1 
	+ \wp{C}{\kGuardB \cdot \ert{C}{0}
		\cdot \sum_{i=0}^{n-1} \wp{C}{\kGuardB}^i} \right. \\
	&\left. \quad \quad \quad \quad \quad \quad 
	{}+ \wp{C}{\kGuardB \cdot \sum\limits_{i=0}^{n-2} \wp{C}{\kGuardB}^i} \right)                           
	\tag{by Equation \ref{eqn:newrule-omega-condition-two}} \\   
	&\eeq  
	1 
	+ \kGuardB \cdot \left( \ert{C}{0} 
	+ 1 
	+ \ert{C}{0} \cdot \sum_{i=1}^{n} \wp{C}{\kGuardB}^i \right. \\
	&\left. \quad \quad \quad \quad \quad \quad 
	{}+ \wp{C}{\kGuardB \cdot \sum\limits_{i=0}^{n-2} \wp{C}{\kGuardB}^i} \right)                           
	\tag{by Equation \ref{eqn:newrule-omega-equation-one}} \\         
	&\eeq  
	1 
	+ \kGuardB \cdot \left( \ert{C}{0} 
	+ 1 
	+ \ert{C}{0} \cdot \sum_{i=1}^{n} \wp{C}{\kGuardB}^i \right. \\
	&\left. \quad \quad \quad \quad \quad \quad 
	{}+ \sum\limits_{i=1}^{n-1} \wp{C}{\kGuardB}^i \right)                           
	\tag{by Equation \ref{eqn:newrule-omega-equation-two}} \\
               &\eeq  1 + \kGuardB \cdot \left( \ert{C}{0} \cdot
              \sum\limits_{i=0}^{n} \wp{C}{\kGuardB}^i ~+~ 
              \sum\limits_{i=0}^{n-1} \wp{C}{\kGuardB}^i \right)~.      
              \tag*{\qed}
       \end{align*}
\end{proof}

\subsection{Proof of Lemma~\ref{lem:apply-rule-relation}}
\label{app:apply-rule-relation}
\subsubsection{Proof of Lemma~\ref{lem:apply-rule-relation}.1}
       %
       %
       %
       %
       %
       %
       %
       %
       %
       %
       %
       %
\begin{proof}
\newcommand{\bxi}[1]{{B_{x_i}}^{#1}}
       We prove a stronger statement: For all sequences $\Blk \in \bnl$ and all $g \in \E$ it holds that 
       \begin{equation}
       \label{eqn:in-proof-lemma-apply-rule-one-to-show}
              \VarsInExp{\wp{\Blk}{g}} ~{}={}~ \VarsInExp{g} \setminus \VarsAssign{\Blk}~.
       \end{equation}       
       It then follows that $\VarsInExp{\wp{\Blk}{g}} \cap \VarsAssign{\Blk} = \emptyset$ and hence $\RelNewRule{\wp{\Blk}{g}}{\Blk}$ by definition.
       Note that every sequence $\Blk \in \bnl$ is of the form $\COMPOSE{\bxi{1}}{\ldots\COMPOSE{}{\bxi{n}}}$  for some $n \geq 1$, where each $\bxi{j}$ is a block.
       Given a program $C \in \pgcl$, let $\VarsGuard{C}$ denote the the set of variables occurring in a guard in $C$. A straightforward induction on the structure of a block
       $B_{x_i} \in \bnl$ yields
       \begin{equation}
              \label{eqn:in-proof-property-block}
              \VarsInExp{\wp{B_{x_i}}{g}} ~{}={}~ \left( \VarsGuard{B_{x_i}} \cup \VarsInExp{g} \right) \setminus \VarsAssign{B_{x_i}}~.
       \end{equation}
       \noindent
       We now proceed by induction on the length of a sequence $n$. \\ \\
       \emph{Induction base}. $\Blk$ consists of a single block $\bxi{1}$. Moreover no guard in $\bxi{1}$ contains a variable since 
       no variable is initialized, i.e.\ $\VarsGuard{C} = \emptyset$. We thus have
       \begin{align*}
              \VarsInExp{\wp{\Blk}{g}} &~{}={}~ \VarsInExp{\wp{{{B_{x_i}}_1}}{g}} \\
              &~{}={}~ \left( \underbracket{\VarsGuard{B_{x_i}}}_{=\emptyset} {}\cup{} \VarsInExp{g} \right) \setminus \VarsAssign{B_{x_i}} 
               \tag{by Equation~\ref{eqn:in-proof-property-block}}\\
              &~{}={}~  \VarsInExp{g} \setminus \VarsAssign{B_{x_i}}~.
       \end{align*} \\
       \emph{Induction hypothesis.} Suppose Equation~\ref{eqn:in-proof-lemma-apply-rule-one-to-show} holds for some arbitrary but fixed sequence $\Blk$ of length $n$ and all $g \in \E$. \\ \\
       \noindent
       \emph{Induction step.} We have 
       \begin{align*}
              &\VarsInExp{\wp{\COMPOSE{\bxi{1}}{\ldots\COMPOSE{\bxi{n}}{\bxi{n+1}}}}{g}} \\
              &\eeq\VarsInExp{\wp{\COMPOSE{\bxi{1}}{\ldots\COMPOSE{}{\bxi{n}}}}{\wp{\bxi{n+1}}{g}}}
              \tag{by Table~\ref{table:wp}} \\
              &\eeq \VarsInExp{\wp{\bxi{n+1}}{g}} \setminus \VarsAssign{\COMPOSE{\bxi{1}}{\ldots\COMPOSE{}{\bxi{n}}}} ~.
              \tag{by I.H.} 
       \end{align*}
       Further, Equation~\ref{eqn:in-proof-property-block} yields 
       \begin{align*}
       \VarsInExp{\wp{\bxi{n+1}}{g}} ~{}={}~ \left( \VarsGuard{\bxi{n+1}} \cup \VarsInExp{g} \right) \setminus \VarsAssign{\bxi{n+1}}~.
       \end{align*}
       Recall that every variable occurring in a guard in $\bxi{n+1}$ must be initialized, that is
       \begin{equation*}
             \VarsGuard{\bxi{n+1}} \subseteq \VarsAssign{\COMPOSE{\bxi{1}}{\ldots\COMPOSE{}{\bxi{n}}}}~.
       \end{equation*}
       Overall, the aforementioned facts and some simple set algebra yields
       \begin{align*}
              &\VarsInExp{\wp{\COMPOSE{\bxi{1}}{\ldots\COMPOSE{\bxi{n}}{\bxi{n+1}}}}{g}} \\
              &\eeq \VarsInExp{g} \setminus \VarsAssign{\COMPOSE{\bxi{1}}{\ldots\COMPOSE{\bxi{n}}{\bxi{n+1}}}}~.
       \end{align*}
\end{proof}
\subsubsection{Proof of Lemma~\ref{lem:apply-rule-relation}.2}
\begin{proof}
       %
       %
       %
       %
       %
       %
       %
       %
       %
       %
       %
       %
       %
       %
       %
       We show that $\VarsInExp{\ert{\Blk}{0}} = \emptyset$ holds for every sequence $\Blk \in \bnl$. This implies $\RelNewRule{\ert{\Blk}{0}}{\Blk}$.
       First observe that since every $\Blk \in \bnl$ is loop--free, we have
       \begin{equation}       
              \label{in-proof-lemma-apply-rule-two-wp-ert}
              \VarsInExp{\wp{\Blk}{g}} = \VarsInExp{\ert{\Blk}{g}}~.
       \end{equation}
       Thus Equation~\ref{eqn:in-proof-lemma-apply-rule-one-to-show} verified in the previous proof yields
       \begin{align*}
              \VarsInExp{\ert{\Blk}{0}} &~{}={}~ \VarsInExp{\wp{\Blk}{0}} 
              \tag{by Equation~\ref{in-proof-lemma-apply-rule-two-wp-ert}}\\
              &~{}={}~ \underbracket{\VarsInExp{0}}_{=\emptyset} \setminus \VarsAssign{\Blk} \\
              \tag{by Equation~\ref{eqn:in-proof-lemma-apply-rule-one-to-show}}
              &~{}={}~ \emptyset~.
       \end{align*}
\end{proof}
\subsection{Proof of Theorem~\ref{thm:soundness-bayesian-network-to-bnl}}
\label{proof-thm:soundness-bayesian-network-to-bnl}
The proof makes use of two technical lemmas that are proven afterwards.

\begin{proof}[Proof of Theorem~\ref{thm:soundness-bayesian-network-to-bnl}]
Assume $\BN$ is a Bayesian network, i.e. $\INPUTS = \emptyset$.
Moreover, let $\sigma \in \Sigma$.
We have to show that
\begin{align*}
        \wp{\TBN{\BN}{\BNCOND}}{\iverson{\bigwedge_{v \in U} \TVAR{v} = \BVAL{v}}}(\sigma)
        \eeq
        \frac{\PROB{\bigwedge_{v \in \NODES} v = \BVAL{v}}}{\PROB{\bigwedge_{v \in O} v = \BVAL{v}}}.
\end{align*}

Then
\begin{align*}
        & \wp{\TBN{\BN}{\BNCOND}}{\iverson{\bigwedge_{v \in U} \TVAR{v} = \BVAL{v}}}(\sigma) \\
        \eeq &
        \frac{
        \wp{\TB{\BN}}{\iverson{\psi} \cdot \iverson{\bigwedge_{v \in U} \TVAR{v} = \BVAL{v}}}(\sigma)
        }{
        \wp{\TB{\BN}}{\iverson{\psi}}(\sigma)
        }
        \tag{Lemma~\ref{lem:bn:loop}, see below} \\
        \eeq &
        \frac{
        \wp{\TB{\BN}}{\iverson{\bigwedge_{v \in \NODES} \TVAR{v} = \BVAL{v}}}(\sigma)
        }{
        \wp{\TB{\BN}}{\iverson{\psi}}(\sigma)
        }
        \tag{Definition of $\psi$} \\
        \eeq &
        \frac{\PROB{\bigwedge_{v \in \NODES} v = \BVAL{v}}}{\PROB{\bigwedge_{v \in O} v = \BVAL{v}}}.
        \tag{Lemma~\ref{lem:bn:structure}, see below}
\end{align*}

\end{proof}

\begin{lemma}\label{lem:bn:loop}
Let $\BN$ be an extended Bayesian network, $\BNCOND$ be a conditioning function, and $f \in \E$.
We define a shortcut for the guard corresponding to all observations: \[ \psi = \bigwedge_{v \in O} \TVAR{v} = \BNCOND(v) \] 
Then, for each $\sigma \in \Sigma$, we have 
\begin{align*}
        \wp{\TBN{\BN}{\BNCOND}}{f}(\sigma) \eeq 
        \frac{
        \wp{\TB{\BN}}{\iverson{\psi} \cdot f}(\sigma)
        }{
        \wp{\TB{\BN}}{\iverson{\psi}}(\sigma)
        }~,
\end{align*}
where, again, we define $\nicefrac 0 0 = 0$.
\end{lemma}

\begin{proof}
\begin{align*}
        &
        \wp{\TBN{\BN}{\BNCOND}}{f} \\
        \eeq &
        \wp{\REPEATUNTIL{\TB{\BN}}{\psi}}{f}
        \tag{Definition of $\TBN{\BN}{\BNCOND}$} \\
        \eeq & 
        \wp{\COMPOSE{\TB{\BN}}{\WHILEDO{\neg \psi}{\TB{\BN}}}}{f}
        \tag{Table~\ref{table:wp}} \\
        \eeq &
        \wp{\TB{\BN}}{\wp{\WHILEDO{\neg \psi}{\TB{\BN}}}{f}}
        \tag{Table~\ref{table:wp}} \\
\end{align*}
We now distinguish two cases: $\wp{\TB{\BN}}{\iverson{\psi}}(\sigma) > 0$ and  $\wp{\TB{\BN}}{\iverson{\psi}}(\sigma) = 0$.

For $\wp{\TB{\BN}}{\iverson{\psi}}(\sigma) > 0$, we proceed as follows:
\begin{align*}
        \eeq &
        \wp{\TB{\BN}}{
            \frac{
                \iverson{\neg \psi} \cdot \wp{\TB{\BN}}{\iverson{\psi} \cdot f}
            }{
                1 - \wp{\TB{\BN}}{\iverson{\neg \psi}}
            }
            + \iverson{\psi} \cdot f 
        }
        \tag{Assumption, Theorem~\ref{thm:newrule-wpsemantics} (1)} \\
        \eeq &
        \wp{\TB{\BN}}{
            \frac{
                \iverson{\neg \psi} \cdot \wp{\TB{\BN}}{\iverson{\psi} \cdot f}
            }{
                1 - \wp{\TB{\BN}}{\iverson{\neg \psi}}
            }
        }
        + \wp{\TB{\BN}}{\iverson{\psi} \cdot f}
        \tag{Theorem~\ref{thm:basic-prop} (2)} \\
        \eeq &
        \wp{\TB{\BN}}{\iverson{\neg \psi}} \cdot 
        \frac{
            \wp{\TB{\BN}}{\iverson{\psi} \cdot f}
        }{
            1 - \wp{\TB{\BN}}{\iverson{\neg \psi}}
        }
        + \wp{\TB{\BN}}{\iverson{\psi} \cdot f}
        \tag{Lemma~\ref{lem:apply-rule-relation} (2) allows to apply Lemma~\ref{lem:newrule-main}} \\
        \eeq &
        \wp{\TB{\BN}}{\iverson{\psi} \cdot f} \cdot \left( 1 + 
        \frac{
            \wp{\TB{\BN}}{\iverson{\neg \psi}}
        }{
            1 - \wp{\TB{\BN}}{\iverson{\neg \psi}}
        }
        \right)
        \tag{Algebra} \\
        \eeq &
        \frac{
            \wp{\TB{\BN}}{\iverson{\psi} \cdot f}
        }{
            1 - \wp{\TB{\BN}}{\iverson{\neg \psi}}
        }
        \tag{Algebra} \\
        \eeq &
        \frac{
            \wp{\TB{\BN}}{\iverson{\psi} \cdot f}
        }{
            \wp{\TB{\BN}}{\iverson{\psi}}
        }
        \tag{Theorem~\ref{thm:basic-prop}}
\end{align*}

For $\wp{\TB{\BN}}{\iverson{\psi}}(\sigma) = 0$, we proceed as follows:
\begin{align*}
        \eeq &
        \wp{\TB{\BN}}{ 
            \iverson{\psi} \cdot f 
        }
        \tag{Assumption, Theorem~\ref{thm:newrule-wpsemantics} (2)} \\
        \eeq &
        0 \\
        \eeq &
        \frac{
            \wp{\TB{\BN}}{\iverson{\psi} \cdot f}
        }{
            \wp{\TB{\BN}}{\iverson{\psi}}
        }
\end{align*}
where the last step is a consequence of our definition $\nicefrac 0 0 = 0$ and the following observations:
\begin{enumerate}
        \item $\wp{\TB{\BN}}{\iverson{\psi}}(\sigma) = 0$ holds by assumption.
        \item $\wp{\TB{\BN}}{\iverson{\psi} \cdot f}(\sigma) = 0$ holds. This shown analogously to the the proof of Theorem~\ref{thm:newrule-wpsemantics}, below equation~(\ref{eq:bennistheorem}).
\end{enumerate}
\end{proof}

\begin{lemma}\label{lem:bn:structure}
Let $\BN$ be an extended Bayesian network, $\BNCOND$ be a conditioning function, and $f \in \E$.
Let $\INPUTS = \{u_1,\ldots,u_k\}$ be the set of inputs of $\BN$.
Moreover, let $\T{u} = (u_1,\ldots,u_k)$ be the canonically ordered tuple of inputs, i.e. $u_1 < u_2 < \ldots < u_k$.
We denote the corresponding tuple of variables by $\TVAR{\T{u}} = (\TVAR{u_1},\ldots,\TVAR{u_k})$. 
Then, for every set $M \subseteq \NODES$, we have
\begin{align*} 
    \wp{\TB{\BN}}{\iverson{\bigwedge_{v \in M} \TVAR{v} = \BVAL{v}}}
    \eeq
    \sum_{\T{y} \in \VALUES^{k}} \iverson{\TVAR{\T{u}} = \T{y}} \cdot \PROB{ \bigwedge_{v \in M} v = \BVAL{v} ~|~ \T{u} = \T{y} }
\end{align*}
\end{lemma}

\begin{proof}
By induction on the number of nodes $n$ of the extended Bayesian network $\BN$.
To improve readability, let $\varphi = \bigwedge_{v \in M} \TVAR{v} = \BVAL{v}$.

\emph{I.B.} For $n = 1$, $\BN$ contains exactly one node, say $w$. 
Let $\VALUES^{|\DEP(w)|} = \{ \T{y_1}, \ldots, \T{y_k} \}$.
\begin{align*}
        & \wp{\TB{\BN}}{\iverson{\varphi}} \\
        \eeq &
        \wp{\TBLOCK{\BN}{w}}{\iverson{\varphi}} \\
        \eeq &
        \sum_{1 \leq j < k} \left( \prod_{1 \leq \ell < j} \iverson{\neg \TGUARD{\BN}{w}{\T{y_{\ell}}}} \right)
        \cdot \iverson{\TGUARD{\BN}{w}{\T{y_{j}}}} \cdot \wp{\TASSIGN{\BN}{w}{\T{y_{j}}}}{\iverson{\varphi}}
        \tag{apply $\wpsymbol$ to $\TBLOCK{\BN}{w}$} \\
        & 
        \quad \,+\, 
        \left( \prod_{1 \leq \ell < k} \iverson{\neg \TGUARD{\BN}{w}{\T{y_{\ell}}}} \right)
        \cdot \wp{\TASSIGN{\BN}{w}{\T{y_{k}}}}{\iverson{\varphi}} \\
        \eeq &
        \sum_{1 \leq j \leq k} \iverson{\T{u} = \T{y_j}} \cdot \wp{\TASSIGN{\BN}{w}{\T{y_{j}}}}{\iverson{\varphi}}
        \tag{guards partition $\VALUES^{|\DEP(w)|}$} \\
        \eeq &
        \sum_{1 \leq j \leq k} \iverson{\T{u} = \T{y_j}} \cdot \CPT{w}(\T{y_j})(\BVAL{w}) \cdot \iverson{\varphi}\subst{\TVAR{w}}{\BVAL{w}}
        \tag{Definition of $\TASSIGN{\BN}{w}{.}$, Table~\ref{table:wp}} \\
        \eeq &
        \sum_{1 \leq j \leq k} \iverson{\T{u} = \T{y_j}} \cdot \PROB{w = \BVAL{w} ~|~ \DEP(w) = \T{y_j}} \cdot \iverson{\varphi}\subst{\TVAR{w}}{\BVAL{w}}
        \tag{Definition of $\CPTSYM$} \\
        \eeq &
        \sum_{1 \leq j \leq k} \iverson{\T{u} = \T{y_j}} \cdot \PROB{w = \BVAL{w} ~|~ \DEP(w) = \T{y_j}}.
        \tag{$\varphi = \true$ or $\varphi = (w = \BVAL{w})$} 
\end{align*}
\emph{I.S.} Now consider an extended Bayesian network $\BN$ whose (smallest) root is $w$.
We distinguish two cases: $w \in M$ and $w \notin M$.
First, assume $w \notin M$. Then
\begin{align*}
        & \wp{\TB{\BN}}{\iverson{\varphi}} \\
        \eeq &
        \wp{\TBLOCK{\BN}{w};\TB{\BN'}}{\iverson{\varphi}}
        \tag{Construction of $\TB{\BN}$} \\
        \eeq &
        \wp{\TBLOCK{\BN}{w}}{\wp{\TB{\BN'}}{\iverson{\varphi}}}
        \tag{Table~\ref{table:wp}} \\
        \eeq &
        \wp{\TBLOCK{\BN}{w}}{
            \underbrace{\sum_{1 \leq j \leq k} \iverson{\T{u} = \T{y_j}} \cdot \PROB{\bigwedge_{v \in M} v = \BVAL{v} ~|~ \T{u} = \T{y_j}}}_{\eeq g}
        }
        \tag{I.H.} \\
        \eeq &
        \sum_{1 \leq j \leq k} \iverson{\T{u} = \T{y_j}} \cdot \PROB{\bigwedge_{v \in M} v = \BVAL{v} ~|~ \T{u} = \T{y_j}}
        \tag{$\wp{\TBLOCK{\BN}{w}}{g} = g$ as $\TBLOCK{\BN}{w}$ does not modify variables in $g$.} \\
\end{align*}
Now assume $w \in M$, and $\VALUES^{|\DEP(w)|} = \{\T{z_1},\ldots,\T{z_{m}}\}$. Then
\begin{align*}
        & \wp{\TB{\BN}}{\iverson{\varphi}} \\
        \eeq &
        \wp{\TBLOCK{\BN}{w};\TB{\BN'}}{\iverson{\varphi}}
        \tag{Construction of $\TB{\BN}$} \\
        \eeq &
        \wp{\TBLOCK{\BN}{w}}{\wp{\TB{\BN'}}{\iverson{\varphi}}}
        \tag{Table~\ref{table:wp}} \\
        \eeq &
        \wp{\TBLOCK{\BN}{w}}{
        \underbrace{\sum_{1 \leq j \leq k', a \in \VALUES} \iverson{w = a \wedge \T{u'} = \T{y_j'}} 
        \cdot \PROB{\bigwedge_{v \in M \setminus \{w\}} v = \BVAL{v} ~|~ w = a, \T{u'} = \T{y_j'}}}_{\eeq g}
        }
        \tag{I.H.} \\
        \eeq &
        \sum_{1 \leq j < k'} \left( \prod_{1 \leq \ell < j} \iverson{\neg \TGUARD{\BN}{w}{\T{z_{\ell}}}} \right)
        \cdot \iverson{\TGUARD{\BN}{w}{\T{z_{j}}}} \cdot \wp{\TASSIGN{\BN}{w}{\T{z_{j}}}}{g}
        \tag{apply $\wpsymbol$ to $\TBLOCK{\BN}{w}$} \\
        & 
        \quad \,+\, 
        \left( \prod_{1 \leq \ell < k'} \iverson{\neg \TGUARD{\BN}{w}{\T{z_{\ell}}}} \right)
        \cdot \wp{\TASSIGN{\BN}{w}{\T{z_{k}}}}{g} \\
        \eeq &
        \sum_{1 \leq j \leq m} \iverson{\T{\TVAR{\DEP(w)}} = \T{z_j}} \cdot \wp{\TASSIGN{\BN}{w}{\T{z_{j}}}}{g}
        \tag{guards partition $\VALUES^{|\INPUTS|}$ and actually all rationals due to $\ELSE$ at the end of each block.} \\
        \eeq &
        \sum_{1 \leq j \leq m} \iverson{\T{\TVAR{\DEP(w)}} = \T{z_j}} \cdot \CPT{w}(\T{z_j})(\BVAL{w}) \cdot g\subst{\TVAR{w}}{\BVAL{w}} 
        \tag{Definition of $\TASSIGN{\BN}{w}{.}$} \\
        \eeq &
        \sum_{1 \leq j \leq m} \iverson{\T{\TVAR{\DEP(w)}} = \T{z_j}} \cdot \PROB{w = \BVAL{w} ~|~ \DEP(w) = \T{z_j}} \cdot g\subst{\TVAR{w}}{\BVAL{w}} 
        \tag{Definition of $\CPT{w}(\T{z_j})(\BVAL{w})$} \\
        \eeq &
        \sum_{1 \leq j \leq m} \iverson{\T{\TVAR{\DEP(w)}} = \T{z_j}} \cdot \PROB{w = \BVAL{w} ~|~ \DEP(w) = \T{z_j}}  
        \tag{Definition of $g$} \\
        & \qquad \cdot \sum_{1 \leq i \leq k'} \iverson{\T{u'} = \T{y_i'}} \cdot \PROB{\bigwedge_{v \in M \setminus \{w\}} v = \BVAL{v} ~|~ w = \BVAL{w}, \T{u'} = \T{y_i'}} 
        \\
        \eeq &
        \sum_{1 \leq j \leq m} \iverson{\T{\TVAR{\DEP(w)}} = \T{z_j}} \sum_{1 \leq i \leq k'} \iverson{\T{u'} = \T{y_i'}} \cdot 
        \tag{Algebra} \\
        & \qquad \PROB{w = \BVAL{w} ~|~ \DEP(w) = \T{z_j}}  
            \cdot \PROB{\bigwedge_{v \in M \setminus \{w\}} v = \BVAL{v} ~|~ w = \BVAL{w}, \T{u'} = \T{y_i'}} \\
        \eeq &
        \sum_{1 \leq j \leq m} \iverson{\T{\TVAR{\DEP(w)}} = \T{z_j}} \sum_{1 \leq i \leq k'} \iverson{\T{u'} = \T{y_i'}}
            \cdot \PROB{\bigwedge_{v \in M} v = \BVAL{v} ~|~ \DEP(w) = \T{z_j}, \T{u'} = \T{y_i'}} \\
        \tag{Chain rule} \\
        \eeq &
        \sum_{1 \leq i \leq k} \iverson{\T{u} = \T{y_i}}
            \cdot \PROB{\bigwedge_{v \in M} v = \BVAL{v} ~|~ \T{u} = \T{y_i}} \\
            \tag{All inputs of $\BN$, i.e. $\T{u}$ are given by $\DEP(w)$ and inputs of $\BN'$, i.e. $\T{u'}$} \\
        \eeq &
        \sum_{\T{y} \in \VALUES^{k}} \iverson{\T{u} = \T{y}}
            \cdot \PROB{\bigwedge_{v \in M} v = \BVAL{v} ~|~ \T{u} = \T{y}}.
\end{align*}
\end{proof}

\end{document}